\documentclass{llncs}
\pdfoutput=1
\usepackage[utf8]{inputenc}
\usepackage{amsmath, amssymb}
\newcommand{\qedhere}{\tag*{$\square$}}

\usepackage{tikz}
\usetikzlibrary{automata,positioning}

\usepackage{graphicx}
\usepackage{tabularx}
\usepackage{csquotes}
\usepackage{url}
\usepackage{booktabs}
\usepackage{color}
\usepackage{pgfplots}
\usepackage{subfig}
\usepackage{afterpage}
\usepackage{capt-of}

\newcommand{\nat}{{\mathbb{N}}}

\newcommand{\ord}{{\sf Ord}}
\newcommand{\ind}{{\sf Ind}}

\newcommand{\dec}{{\sf Dec}}
\newcommand{\acc}{{\sf Acc}}

\tikzset{
    state/.style={
		        rectangle,
            rounded corners,
            draw=black,
            minimum height=2em,
            minimum width=2em,
            align=center,
            }
}

\tikzset{
    statep/.style={
            circle,
            draw=black,
            minimum height=2em,
            minimum width=2em,
            align=center,
            }
}

\tikzstyle{acc}=[double]

\newcommand{\X}{{\ensuremath{\mathbf{X}}}}
\newcommand{\F}{{\ensuremath{\mathbf{F}}}}
\newcommand{\G}{{\ensuremath{\mathbf{G}}}}
\newcommand{\U}{{\ensuremath{\mathbf{U}}}}

\newcommand{\true}{{\ensuremath{\mathbf{tt}}}}
\newcommand{\false}{{\ensuremath{\mathbf{ff}}}}

\newcommand{\setG}{\mathcal{G}}
\newcommand{\gsf}{{\ensuremath{\mathbb{G}}}}
\newcommand{\Reach}{{\it Reach}}

\newcommand{\oracledet}{\mathcal{U}}
\newcommand{\prodaut}{\mathcal{P}}

\newcommand{\compaut}{\mathcal{A}}

\newcommand{\aft}{{\it af}}

\newcommand{\os}{0.5mm}

\newcommand{\A}{\mathcal A}

\newcommand{\lang}{\mathsf{L}}
\newcommand{\width}{\mathit{width}}

\newcommand{\theoremlike}[2]{\par\medskip\penalty-250\refstepcounter{theorem}{{\bfseries\noindent#2
			\ref{#1}.}}}

\newcommand{\thmhelperpre}[2]{\theoremlike{#1}{#2}}
\newcommand{\thmhelperpost}{\par\medskip}

\newenvironment{reftheorem}[1]{\thmhelperpre{#1}{Theorem}}{\thmhelperpost}

\newenvironment{refproposition}[1]{\thmhelperpre{#1}{Proposition}}{
	\thmhelperpost }

\newcommand{\myspace}{\vspace*{-1em}}
\newcommand{\myspacet}{\vspace*{-0.3em}}

\begin{document}

\advance\textheight4mm
\advance\voffset-2mm

\title{From LTL and Limit-Deterministic B\"uchi Automata to Deterministic Parity Automata
\thanks{This work is partially funded by the DFG Research Training Group \enquote{PUMA: Programm- und Modell-Analyse} (GRK 1480), DFG project \enquote{Verified Model Checkers}, the ERC Starting Grant (279499: inVEST), and the Czech Science Foundation, grant No.~\mbox{P202/12/G061}.}
}
\author{Javier Esparza\inst{1} \and Jan K\v{r}et{\'i}nsk{\'y}\inst{1} \and Jean-Fran\c{c}ois Raskin \inst{2} \and Salomon Sickert \inst{1} }

\institute{Technische Universität München \email{\{esparza, jan.kretinsky, sickert\}@in.tum.de} \and Université libre de Bruxelles \email{jraskin@ulb.ac.be}}    
           
\maketitle

\pagestyle{plain}

\myspace

\begin{abstract}
Controller synthesis for general linear temporal logic (LTL) objectives is a challenging task. The standard approach involves translating the LTL objective into a deterministic parity automaton (DPA) by means of the Safra-Piterman construction. One of the challenges is the size of the DPA, which often grows very fast in practice, and can reach double exponential size in the length of the LTL formula. In this paper we describe a single exponential translation from limit-deterministic B\"uchi automata (LDBA) to DPA, and show that it can be concatenated with a recent efficient translation from LTL to LDBA to yield a double exponential, \enquote{Safraless} LTL-to-DPA construction. We also report on an implementation, a comparison with the SPOT library, and performance on several sets of formulas, including instances from the 2016 SyntComp competition. 
\end{abstract}


\section{Introduction}

Limit-deterministic Büchi automata (LDBA, also known as semi-deterministic Büchi automata) were introduced by Courcoubetis and Yannakakis (based on previous work by Vardi) to solve the 
qualitative probabilistic model-checking problem: Decide if the executions of a Markov chain or Markov Decision Process satisfy a given LTL formula with probability 1 \cite{DBLP:conf/focs/Vardi85,DBLP:conf/lics/VardiW86,DBLP:journals/jacm/CourcoubetisY95}. The problem faced by these authors was that fully nondeterministic Büchi automata (NBAs), which are as expressible as LTL, and more, cannot be used for probabilistic model checking, and deterministic Büchi automata (DBA) are less expressive than LTL. The solution was to introduce LDBAs as a model in-between: as expressive as NBAs, but deterministic enough.

After these papers, LDBAs received little attention. The alternative path of translating the LTL formula into an equivalent fully deterministic Rabin automaton using Safra's construction \cite{DBLP:conf/focs/Safra88} was considered a better option, mostly because it also solves the quantitative probabilistic model-checking problem (computing the probability of the executions that satisfy a formula). However, recent papers  have shown that LDBAs were unjustly forgotten. Blahoudek {\em et al.} have shown that LDBAs are easy to complement \cite{DBLP:conf/tacas/BlahoudekHSST16}. Kini and Viswanathan have given a single exponential translation of
LTL$_{\setminus\G\U}$ to LDBA \cite{DBLP:conf/tacas/KiniV15}. Finally, 
Sickert {\em et al.} describe in \cite{DBLP:conf/cav/SickertEJK16}  a double exponential translation for full LTL that can also be applied to the quantitative case, and behaves better than Safra's construction in practice.

In this paper we add to this trend by showing that LDBAs are also attractive for synthesis. The standard solution to the synthesis problem with LTL objectives consists of translating the LTL formula into a deterministic parity automaton (DPA) with the help of the Safra-Piterman construction \cite{DBLP:journals/lmcs/Piterman07}. While limit-determinism is not ``deterministic enough'' for the synthesis problem, we introduce a conceptually simple and worst-case optimal translation LDBA$\rightarrow$DPA. Our translation bears some similarities with that of \cite{bernd} where, 
however, a Muller acceptance condition is used. This condition can also be phrased as a Rabin condition, but not as a parity condition. Moreover, the way of tracking all possible states and finite runs differs.

Together with the translation LTL$\rightarrow$LDBA of \cite{DBLP:conf/cav/SickertEJK16}, our construction provides a \enquote{Safraless}, procedure to obtain a DPA from an LTL formula. However, the direct concatenation of the two constructions does not yield an algorithm of optimal complexity: the LTL$\rightarrow$LDBA translation is double exponential (and there is a double-exponential lower bound), and so for the LTL$\rightarrow$DPA translation we only obtain a triple exponential bound. In the second part of the paper we solve this problem. We show that the LDBAs derived from LTL formulas satisfy a special property, and prove that for such automata the concatenation of the two constructions remains double exponential. To the best of our knowledge, this is the first 
double exponential \enquote{Safraless} LTL$\rightarrow$DPA procedure. (Another asymptotically optimal
\enquote{Safraless} procedure for determinization of B\"uchi automata with Rabin automata as target
has been presented in \cite{DBLP:journals/iandc/FogartyKVW15}.)

In the third and final part, we report on the performance of an 
implementation of our LTL$\rightarrow$LDBA$\rightarrow$DPA construction, and compare it with
algorithms implemented in the SPOT library \cite{duret.16.atva2}. 
Note that it is not possible to force SPOT to always produce DPA, sometimes it produces a deterministic generalized Büchi automaton (DGBA).
The reason is that DGBA are often smaller than DPA (if they exist) and game-solving algorithms for DGBA are not less efficient than for DPA. 
Therefore, also our implementation may produce DGBA in some cases.
We show that our implementation outperforms SPOT for several sets of parametric formulas and formulas used in synthesis examples taken from the SyntComp 2016 competition, and remains competitive  for randomly generated formulas.

\myspace

\paragraph{{\bf Structure of the paper}}
Section~2 introduces the necessary preliminaries about automata. Section~3 defines the translation LDBA$\rightarrow$DPA. Section~4 shows how to compose of LTL$\rightarrow$LDBA  and LDBA$\rightarrow$DPA in such a way that the resulting DPA is at most doubly exponential in the size of the LTL formula. Section~5 reports on the experimental evaluation of this worst-case optimal translation, and Section~6 contains our conclusions.
The paper is a full version of a TACAS'17 paper.


\myspacet

\section{Preliminaries}

\myspacet

\paragraph{{\bf B\"uchi automata}} A (nondeterministic) $\omega$-word automaton $A$ with B\"uchi acceptance condition (NBA) is a tuple $(Q,q_0,\Sigma,\delta,\alpha)$ where $Q$ is a finite set of states, $q_0 \in Q$ is the {\em initial} state, $\Sigma$ is a finite alphabet, $\delta \subseteq Q \times \Sigma \times Q$ is the transition relation, and $\alpha \subseteq \delta$ is the set of {\em accepting} transitions\footnote{Here, we consider automata on infinite words with acceptance conditions based on transitions. It is well known that there are linear translations from automata with acceptance conditions defined on transitions to automata with acceptance conditions defined on states, and vice-versa.}. W.l.o.g. we assume that $\delta$ is total in the following sense: for all $q \in Q$, for all $\sigma \in \Sigma$, there exists $q' \in Q$ such that $(q,\sigma,q') \in \delta$. $A$ is {\em deterministic} if for all $q \in Q$, for all $\sigma \in \Sigma$, there exists a unique $q' \in Q$ such that $(q,\sigma,q') \in \delta$. When $\delta$ is deterministic and total, it can be equivalently seen as a function $\delta : Q \times \Sigma \rightarrow Q$. Given $S \subseteq Q$ and $\sigma \in \Sigma$, let ${\sf post}^{\sigma}_{\delta}(S)=\{ q' \mid \exists q \in S \cdot (q,\sigma,q') \in \delta \}$.

A {\em run} of $A$ on a $\omega$-word $w : \nat \rightarrow \Sigma$ is a $\omega$-sequence of states $\rho : \nat \rightarrow Q$ such that $\rho(0)=q_0$ and for all positions $i \in \nat$, we have that $(\rho(i),w(i),\rho(i+1)) \in \delta$. A run $\rho$ is {\em accepting} if there are infinitely many positions $i \in \nat$ such that $(\rho(i),w(i),\rho(i+1)) \in \alpha$.  The {\em language} defined by $A$, denoted by $\lang(A)$, is the set of $\omega$-words $w$ for which $A$ has an accepting run.

A {\em limit-deterministic B\"uchi automaton} (LDBA) is a B\"uchi automaton $A=(Q,q_0,\Sigma,\delta,\alpha)$ such that there exists a subset $Q_d \subseteq Q$ satisfying the three following properties:
  \begin{enumerate}
  	\item $\alpha \subseteq Q_d \times \Sigma \times Q_d$, i.e. all accepting transitions are transitions within $Q_d$;
	\item $\forall q \in Q_d \cdot \forall \sigma \in \Sigma \cdot \forall q_1,q_2 \in Q \cdot (q,\sigma,q_1) \in \delta \land (q,\sigma,q_2) \in \delta \rightarrow q_1=q_2$, i.e. the transition relation $\delta$ is deterministic within $Q_d$
	\item $\forall q \in Q_d \cdot \forall \sigma \in \Sigma \cdot \forall q' \in Q \cdot (q,\sigma,q') \in \delta \rightarrow q' \in Q_d$, i.e. $Q_d$ is a trap (when $Q_d$ is entered it is never left).
  \end{enumerate}
  \noindent
 W.l.o.g. we assume that $q_0 \in Q \setminus Q_d$, and we denote $Q \setminus Q_d$ by $\overline{Q_d}$.
  Courcoubetis and Yannakakis show that for every $\omega$-regular language ${\cal L}$, there exists an LDBA $A$  such that $\lang(A)={\cal L}$ \cite{DBLP:journals/jacm/CourcoubetisY95}. That is, LDBAs are as expressive as NBAs. 
  An example of LDBA is given in Fig.~\ref{fig:ex-LDBA}. Note that the language accepted by this LDBA cannot be recognized by a deterministic B\"uchi automaton.

\begin{figure}[t]
	\myspace
  \begin{center}

  \begin{tikzpicture}
  [x=2cm,y=4mm,font=\small,initial text=,outer sep=0pt]
  \tikzset{
  	state/.style={
  		circle,
  		rounded corners,
  		draw=black,
  		minimum size=1em,
  		align=center,
  	}
  }
  
  \node[state,initial] (1) at (0,0) {$1$};
  \node[state] (2) at (1,1) {$2$}; 
  \node[state] (3) at (1,-1) {$3$};   
  \node[state] (4) at (2,0) {$4$};   		
  
  \path[->]
  (1) edge[] node[above]{$a$}   (2)
  (1) edge[loop above] node[above]{$\Sigma$}  ()
  (2) edge[loop above,ultra thick] node[left,pos=0.2]{$a$}  ()
  (2) edge[loop above,white] ()
  (1) edge[]  node[below]{$b$}  (3)
  (3) edge[loop below,ultra thick]  node[left,pos=0.8]{$b$} ()
  (3) edge[loop below, white] ()
  (2) edge[] node[above]{$\Sigma\setminus\{a\}$}  (4.165)
  (3) edge[] node[below]{$\Sigma\setminus\{b\}$}   (4.195)   
  (4) edge[loop above] node[above]{$\Sigma$}  ()	
  ;
  \end{tikzpicture}
    \end{center}
    \myspace\myspace
\caption{\label{fig:ex-LDBA} An LDBA for the LTL language $\F \G a \lor \F \G b$. The behavior of $A$ is deterministic within the subset of states $Q_d=\{2,3,4\}$ which is a trap, the set of accepting transitions are depicted in bold face and they are defined only between states of $Q_d$.}
\myspace
\end{figure}
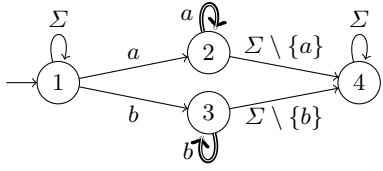

\paragraph{{\bf Parity automata}}
A  deterministic $\omega$-word automaton $A$ with {\em parity} acceptance condition (DPA) is a tuple $(Q,q_0,\Sigma,\delta,p)$, defined as for deterministic B\"uchi automata with the exception of the acceptance condition $p$, which is now a function assigning an integer in $\{ 1, 2, \dots, d \}$, called a {\em color}, to each transition in the automaton. Colors are naturally ordered by the order on integers. 

Given a run $\rho$ over a word $w$, the infinite sequence of colors traversed by the run $\rho$ is noted $p(\rho)$ and is equal to $p(\rho(0),w(0),\rho(1))$ $p((\rho(1),w(1),\rho(2)) \dots$ $p(\rho(n),w(n),\rho(n+1)) \dots$. A run $\rho$ is {\em accepting} if the minimal color that appears infinitely often along $p(\rho)$ is {\em even}. The {\em language} defined by $A$, denoted by $\lang(A)$ is the set of $\omega$-words $w$ for which $A$ has an accepting run.

While deterministic B\"uchi automata are not expressively complete for the class of $\omega$-regular languages, DPAs are complete for $\omega$-regular languages: for every $\omega$-regular language ${\cal L}$ there exists a DPA $A$ such that $\lang(A)={\cal L}$, see e.g.~\cite{DBLP:journals/lmcs/Piterman07}.

\section{From LDBA to DPA}

\subsection{Run DAGs and their coloring}

\paragraph{{\bf Run DAG}} A nondeterministic automaton $A$ may have several (even an infinite number of) runs on a given $\omega$-word $w$. As in~\cite{DBLP:journals/tocl/KupfermanV01}, we represent this set of runs by means of a directed acyclic graph structure called the {\em run DAG} of $A$ on $w$. Given an LDBA $A=(Q,Q_d,q_0,\Sigma,\delta,\alpha)$, this graph $G_w=(V,E)$ has a set of vertices $V \subseteq Q \times \nat$ and edges $E \subseteq V \times V$ defined as follows:

  \begin{itemize}
  	\item $V = \bigcup_{i \in \nat} V_i$, where the sets $V_i$ are defined inductively: 
			\begin{itemize}
				\item $V_0=\{ (q_0,0) \}$, and for all $i \geq 1$, 
				\item $V_i = \{ (q,i) \mid \exists (q',i-1) \in V_{i-1} : (q',w(i),q) \in \delta \}$;
			\end{itemize}
	\item $E = \{ ((q,i),(q',i+1)) \in V_i \times V_{i+1} \mid (q,w(i),q') \in \delta \}$.
  \end{itemize} 
  \noindent
  We denote by $V^d_{i}$ the set $V_i \cap (Q_d \times \{i\})$ that contains the subset of vertices of layer $i$ that are associated with states in $Q_d$.

Observe that all the paths of $G_w$ that start from $(q_0,0)$ are runs of $A$ on $w$, and, conversely,  each run $\rho$ of $A$ on $w$ corresponds exactly to one path in $G_w$ that starts from $(q_0,0)$. So, we call {\em runs} the paths in the run DAG $G_w$. In particular, we say that an infinite path $v_0 v_1 \dots v_n \dots$ of $G_w$ is an accepting run if there are infinitely many positions $i \in \nat$ such that $v_i=(q,i)$, $v_{i+1}=(q',i+1)$, and $(q,w(i),q') \in \alpha$. Clearly, $w$ is accepted by $A$ if and only if there is an accepting run in $G_w$.  We denote by $\rho(0..n)=v_0 v_1 \dots v_n$ the prefix of length $n+1$ of the run $\rho$.

\paragraph{{\bf Ordering of runs}} 
A function $\ord : Q \rightarrow \{ 1,2,\dots,|Q_d|,+\infty\}$ is called an {\em ordering} of the states of $A$ w.r.t. $Q_d$ if  $\ord$ defines a strict total order on the state from $Q_d$, and maps each state $q \in \overline{Q_d}$ to $+ \infty$, i.e.: 
  \begin{itemize}
	\item for all $q \in \overline{Q_d}$, $\ord(q)=+\infty$, 
	\item for all $q \in Q_d$, $\ord(q)\not=+\infty$, and 
	\item for all $q,q' \in Q_d$, $\ord(q)=\ord(q')$ implies $q=q'$. 
  \end{itemize}
  \noindent
  We extend $\ord$ to vertices in $G_w$ as follows: $\ord((q,i))=\ord(q)$.

Starting from $\ord$, we define the following pre-order on the set of run prefixes of the run DAG $G_w$.
Let $\rho(0..n)=v_0 v_1 \dots v_n \dots$ and $\rho'(0..n)=v'_0 v'_1 \dots v'_n \dots$ be two run prefixes of length $n+1$, we write $\rho(0..n) \sqsubseteq \rho'(0..n)$, if $\rho(0..n)$ is {\em smaller than} $\rho'(0..n)$, which is defined as:
  \begin{itemize}
  	\item for all $i$, $0 \leq i \leq n$, $\ord(\rho(i))=\ord(\rho'(i))$, or 
	\item there exists $i$, $0 \leq i \leq n$, such that:
		\begin{itemize}
			\item $\ord(\rho(i)) < \ord(\rho'(i))$, and 
			\item for all $j$, $0 \leq j < i$, $\ord(\rho(j))=\ord(\rho'(j))$. 
		\end{itemize}
  \end{itemize}
  \noindent
  This is extended to (infinite) runs as: $\rho \sqsubseteq \rho'$ iff for all $i \geq 0 \cdot \ord(\rho(0..i)) \sqsubseteq \ord(\rho'(0..i))$. 

\begin{remark}
If $A$ accepts a word $w$, then $A$ has a $\sqsubseteq$-smallest accepting run for $w$.
\end{remark}

We use the $\sqsubseteq$-relation on run prefixes to order the vertices of $V_i$ that belong to $Q_d$: for two different vertices $v=(q,i) \in V_i$ and $v'=(q',i) \in V_i$, $v$ is $\sqsubset_i$-smaller than $v'$, if there is a run prefix of $G_w$ that ends up in $v$ which is $\sqsubseteq$-smaller than all the run prefixes that ends up in $v'$, which induces a total order among the vertices of $V^d_i$ because the states in $Q_d$ are totally ordered by the function $\ord$. 

\begin{lemma}
For all $i \geq 0$, for two different vertices $v=(q,i),v'=(q',i) \in V^d_i$, then either $v \sqsubset_i v'$ or $v' \sqsubset_i v$, i.e., $\sqsubset_i$ is a total order on $V^d_i$.
\end{lemma}

\paragraph{{\bf Indexing vertices}}
The index of a vertex $v=(q,i) \in V_i$ such that $q \in Q_d$, denoted by $\ind_i(v)$, is a value in $\{1,2,\dots,|Q_d|\}$ that denotes its order in $V^d_i$ according to $\sqsubset_i$ (the $\sqsubset_i$-smallest element has index $1$). 
For $i \geq 0$, we identify two important sets of vertices:
  \begin{itemize}
  	\item $\dec(V^d_{i})$ is the set of vertices $v \in V^d_{i}$ such that there exists a vertex $v' \in V^d_{i+1}$: $(v,v') \in E$ and $\ind_{i+1}(v') < \ind_{i}(v)$, i.e. the set of vertices in $V^d_{i}$ whose (unique) successor in $V^d_{i+1}$ has a smaller index value. 
	\item $\acc(V^d_{i})$ is the set of vertices $v=(q,i) \in V^d_{i}$ such that there exists $v'=(q',i+1) \in V^d_{i+1}$: $(v,v') \in E$ and $(q,w(i),q') \in \alpha$, i.e. the set of vertices in $V^d_{i}$ that are the source of an accepting transition on $w(i)$.
  \end{itemize}

\begin{remark}
Along a run, the index of vertices can only decrease. As the function $\ind(\cdot)$ has a finite range, the index along a run has to eventually stabilize.
\end{remark}

\paragraph{{\bf Assigning colors}}
The set of colors that are used for coloring the levels of the run DAG $G_w$ is $\{1, 2, \dots, 2\cdot|Q_d|+1\}$.
We associate a color with each transition from level $i$ to level $i+1$ according to the following set of cases:
\begin{enumerate}
	\item if $\dec(V^d_{i})=\emptyset$ and $\acc(V^d_i)\not=\emptyset$, the color is $2 \cdot\min_{v \in \acc(V^d_{i})} \ind_{i}(v)$.
	\item if $\dec(V^d_{i})\not=\emptyset$ and $\acc(V^d_i)=\emptyset$, the color is $2 \cdot\min_{v \in \dec(V^d_{i})} \ind_{i}(v)-1$.
	\item if $\dec(V^d_{i})\not=\emptyset$ and $\acc(V^d_i)\not=\emptyset$, the color is defined as the minimal color among 
		\begin{itemize}
			\item $c_{{\sf odd}}=2 \cdot\min_{v \in \dec(V^d_{i})} \ind_{i}(v)-1$, and 
			\item $c_{{\sf even}}=2 \cdot \min_{v \in \acc(V^d_{i})} \ind_{i}(v)$.
		\end{itemize}
	\item if $\dec(V^d_{i})=\acc(V^d_i)=\emptyset$, the color is $2 \cdot |Q_q|+1$.
\end{enumerate}

The intuition behind this coloring is as follows: the coloring tracks runs in $Q_d$ (only those are potentially accepting as $\alpha \subseteq Q_d \times \Sigma \times Q_d$) and tries to produce an even color that corresponds to the smallest index of an accepting run. If in level $i$ the run DAG has an outgoing transition that is accepting, then this is a \textit{positive event}, as a consequence the color emitted is {\em even} and it is a function of the smallest index of a vertex associated with an accepting transition from $V_{i}$ to $V_{i+1}$. Runs in $Q_d$ are deterministic but they can merge with \textit{smaller} runs. When this happens, this is considered as a \textit{negative event} because the even colors that have been emitted by the run that merges with the smaller run should  not be taken into account anymore. As a consequence an odd color is emitted in order to cancel all the (good) even colors that were generated by the run that merges with the smaller one. In that case the odd color is function of the smallest index of a run vertex in $V_{i}$ whose run merges with a smaller vertex in $V_{i+1}$. Those two first cases are handled by cases $1$ and $2$ of the case study above. When both situations happen at the same time, then the color is determined by the minimum of the two colors assigned to the positive and the negative events. This is handled by case 3 above. And finally, when there is no accepting transition from $V_{i}$ to $V_{i+1}$ and no merging, the largest odd color is emitted as indicated by case 4 above. 

According to this intuition, we define the {\em color summary} of the run DAG $G_w$ as the minimal color that appears infinitely often along the transitions between its levels.  Because of the deterministic behavior of the automaton in $Q_d$, each run can only merge at most $| Q_d |-1$ times with a smaller one (the size of the range of the function $\ind(\cdot)$ minus one), and as a consequence of the definition of the above coloring, we know that, on word accepted by $A$, the smallest accepting run will eventually generate infinitely many (good) even colors that are never trumped by smaller odd colors. 
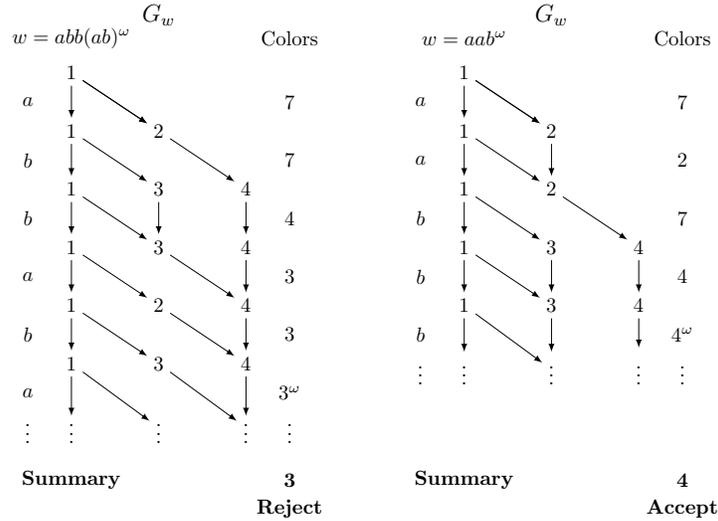
\begin{figure}
	\myspace\myspace
 \begin{center}
        \resizebox{!}{7cm}{\begin{tikzpicture}
[x=1.5cm,y=1.0cm,font=\normalsize,initial text=,outer sep=0pt]

\begin{scope}
\node (w0) at (-0.5,-0.5) {$a$};
\node (w1) at (-0.5,-1.5) {$b$};
\node (w2) at (-0.5,-2.5) {$b$};
\node (w3) at (-0.5,-3.5) {$a$};
\node (w4) at (-0.5,-4.5) {$b$};
\node (w4) at (-0.5,-5.5) {$a$};
\node (w5) at (-0.5,-6.1) {$\vdots$};

\node (w)  at (0,0.6) {$w=abb(ab)^\omega$};
\node (00) at (0,0) {$1$};
\node (01) at (0,-1) {$1$};
\node (02) at (0,-2) {$1$};
\node (03) at (0,-3) {$1$};
\node (04) at (0,-4) {$1$};
\node (05) at (0,-5) {$1$};
\node (06) at (0,-6.1) {$\vdots$};
\node (07) at (0,-7) {\textbf{Summary}};

\node (Gw)  at (1,1.0) {{\large $G_w$}};
\node (11) at (1,-1) {$2$};
\node (12) at (1,-2) {$3$};
\node (13) at (1,-3) {$3$};
\node (14) at (1,-4) {$2$};
\node (15) at (1,-5) {$3$};
\node (16) at (1,-6.1) {$\vdots$};

\node (22) at (2,-2) {$4$};
\node (23) at (2,-3) {$4$};
\node (24) at (2,-4) {$4$};
\node (25) at (2,-5) {$4$};
\node (26) at (2,-6.1) {$\vdots$};

\node (c)  at (2.5,0.6) {Colors};
\node (30) at (2.5,-0.5) {$7$};
\node (31) at (2.5,-1.5) {$7$};
\node (32) at (2.5,-2.5) {$4$};
\node (33) at (2.5,-3.5) {$3$};
\node (34) at (2.5,-4.5) {$3$};
\node (35) at (2.5,-5.5) {$3^\omega$};
\node (36) at (2.5,-6.1) {$\vdots$};
\node (37) at (2.5,-7) {$\mathbf{3}$};
\node (38) at (2.5,-7.5) {\textbf{Reject}};

\pgfsetarrowsend{latex}

\draw (00)--(01);\draw (01)--(02); \draw (02)--(03);\draw (03)--(04);\draw (04)--(05);\draw (05)--(0,-5.9);
\draw (12)--(13);
\draw (22)--(23);\draw (23)--(24);\draw (24)--(25);\draw (25)--(2,-5.9);
\draw (00)--(11);

\draw (00)--(11);\draw (01)--(12); \draw (02)--(13);\draw (03)--(14);\draw (04)--(15);\draw (05)--(16);
\draw (11)--(22);\draw (13)--(24); \draw (14)--(25);\draw (15)--(26);
\end{scope}

\begin{scope}[shift={(4.5,0)}]
\node (w0) at (-0.5,-0.5) {$a$};
\node (w1) at (-0.5,-1.5) {$a$};
\node (w2) at (-0.5,-2.5) {$b$};
\node (w3) at (-0.5,-3.5) {$b$};
\node (w3) at (-0.5,-4.5) {$b$};
\node (w5) at (-0.5,-5.1) {$\vdots$};

\node (w)  at (0,0.6) {$w=aab^\omega$};
\node (00) at (0,0) {$1$};
\node (01) at (0,-1) {$1$};
\node (02) at (0,-2) {$1$};
\node (03) at (0,-3) {$1$};
\node (04) at (0,-4) {$1$};
\node (05) at (0,-5.1) {$\vdots$};
\node (07) at (0,-7) {\textbf{Summary}};

\node (Gw)  at (1,1.0) {{\large $G_w$}};
\node (11) at (1,-1) {$2$};
\node (12) at (1,-2) {$2$};
\node (13) at (1,-3) {$3$};
\node (14) at (1,-4) {$3$};
\node (15) at (1,-5.1) {$\vdots$};

\node (23) at (2,-3) {$4$};
\node (24) at (2,-4) {$4$};
\node (25) at (2,-5.1) {$\vdots$};

\node (c)  at (2.5,0.6) {Colors};
\node (30) at (2.5,-0.5) {$7$};
\node (31) at (2.5,-1.5) {$2$};
\node (32) at (2.5,-2.5) {$7$};
\node (33) at (2.5,-3.5) {$4$};
\node (35) at (2.5,-4.5) {$4^\omega$};
\node (36) at (2.5,-5.1) {$\vdots$};
\node (37) at (2.5,-7) {$\mathbf{4}$};
\node (38) at (2.5,-7.5) {\textbf{Accept}};

\pgfsetarrowsend{latex}

\draw (00)--(01);\draw (01)--(02); \draw (02)--(03);\draw (03)--(04);\draw (04)--(0, -4.8);
\draw (00)--(11);\draw (01)--(12); \draw (02)--(13);\draw (03)--(14);\draw (04)--(15);
\draw (11)--(12);\draw (13)--(14);\draw (14)--(1,-4.8);
\draw (12)--(23);
\draw (23)--(24);\draw (24)--(25);
\draw (00)--(11);

\
\end{scope}

\end{tikzpicture}}
  \end{center}
  \myspace
\caption{\label{fig:run-dag} The run DAGs automaton of Fig.~\ref{fig:ex-LDBA} on the word $w=(ab)^{\omega}$ given on the left, and on the word $w=aab^{\omega}$ given on the right, together with their colorings.}
\myspace\myspace
\end{figure}

\begin{example}
The left part of Fig.~\ref{fig:run-dag} depicts the run DAG of the limit-deterministic automaton of Fig.~\ref{fig:ex-LDBA} on the word $w=abb(ab)^{\omega}$. Each path in this graph represents a run of the automaton on this word. The coloring of the run DAG follows the coloring rules defined above. Between level $0$ and level $1$, the color is equal to $7= 2|Q_d| + 1$, as no accepting edge is taken from level $0$ to level $1$ and no run merges (within $Q_d$). The color $7$ is also emitted from level $1$ to level $2$ for the same reason. The color $4$ is emitted from level $2$ to level $3$ because the accepting edge $(3,b,3)$ is taken and the index of state $3$ in level $2$ is equal to $2$ (state $4$ has index $1$ as it is the end point of the smallest run prefix within $Q_d$). The color $3$ is emitted from level $3$ to level $4$ because the run that goes from $3$ to $4$ merges with the smaller run that goes from $4$ to $4$. In order to cancel the even colors emitted by the run that goes from $3$ to $4$, color $3$ is emitted. It cancels the even color $4$ emitted before by this run. Afterwards, colors $3$ is emitted forever. The color summary is $3$ showing that there is no accepting run in the run DAG.

The right part of Fig.~\ref{fig:run-dag} depicts the run DAG of the limit deterministic automaton of Fig.~\ref{fig:ex-LDBA} on the word $w=aab^{\omega}$. The coloring of the run DAG follows the coloring rules defined above. Between levels $0$ and $1$, color $7$ is emitted because no accepting edge is crossed. To the next level, we see the accepting edge $(2,a,2)$ and color $2\cdot1=2$ is emitted.
Upon reading the first $b$, we see again $7$ since there is neither any accepting edge seen nor any merging takes place. 
Afterwards, each $b$ causes an accepting edge $(3,b,3)$ to be taken. While the smallest run, which visits $4$ forever, is not accepting, the second smallest run that visits $3$ forever is accepting. As $3$ has index $2$ in all the levels below level $3$, the color is forever equal to $4$. The color summary of the run is thus equal to $2\cdot 2=4$ and this shows that word $w=aab^{\omega}$ is accepted by our limit deterministic automaton of Fig.~\ref{fig:ex-LDBA}.
\end{example}

The following theorem tells us that the color summary (the minimal color that appears infinitely often) can be used to identify run DAGs that contain accepting runs. The proof can be found in Appendix~\ref{app:coloring}.

\begin{theorem} \label{thm:coloring}
The color summary of the run DAG $G_w$ is even if and only if there is an accepting run in $G_w$.
\end{theorem}


\subsection{Construction of the DPA}

From an LDBA $A=(Q,Q_d,q_0,\Sigma,\delta,\alpha)$ and an ordering function $\ord  : Q \rightarrow \{1,2,\dots,|Q_d|,+\infty\}$ compatible with $Q_d$, we construct a deterministic parity automaton $B=(Q^B,q_0^B,\Sigma,\delta^B,p)$ that, on a word $w$, constructs the levels of the run DAG $G_w$ and the coloring of previous section. Theorem~\ref{thm:coloring} tells us that such an automaton accepts the same language as $A$. 

First, we need some notations. Given a finite set $S$, we note ${\cal P}(S)$ the set of its subsets, and ${\cal OP}(S)$ the set of its totally ordered subsets. So if $(s,<) \in {\cal OP}(S)$ then $s \subseteq S$ and $\mathord{<} \subseteq s \times s$ is a total strict order on $s$. For $e \in s$, we denote by ${\sf Ind}_{(s,<)}(e)$ the position of $e \in s$ among the elements in $s$ for the total strict order $<$, with the convention that the index of the $<$-minimum element is equal to $1$. The deterministic parity automaton $B=(Q^B,q_0^B,\Sigma,\delta^B,p)$ is defined as follows.

\paragraph{{\bf States and initial state}}
The set of states is $Q^B={\cal P}(\overline{Q_d}) \times {\cal OP}(Q_d)$, i.e. a state  of $B$ is a pair $(s,(t,<))$ where $s$ is a set of states outside $Q_d$, and $t$ is an ordered subset of $Q_d$. The ordering reflects the relative index of each state within $t$. The initial state is $q^B_0=(\{q_0\},(\{\},\{\}))$.

\paragraph{{\bf Transition function}}
Let $(s_1,(t_1,<_1))$ be a state in $Q^B$, and $\sigma \in \Sigma$. Then $\delta^B((s_1,(t_1,<_1)))=(s_2,(t_2,<_2))$ where:
		\begin{itemize}
			\item $s_2 = {\sf post}^{\sigma}_{\delta}(s_1) \cap \overline{Q_d}$;
			\item $t_2 =  {\sf post}^{\sigma}_{\delta}(s_1 \cup t_1) \cap Q_d$;
			\item $<_2$ is defined from $<_1$ and ${\sf Ord}$ as follows:
				$\forall q_1,q_2 \in t_2$: $q_1 <_2 q_2$ iff:
			
			  \begin{enumerate}
			  	\item {\bf either}, $\neg\exists q'_1 \in t_1:q_1=\delta(q'_1,\sigma)$, and $\neg\exists q'_2 \in t_1:q_2=\delta(q_2',\sigma)$, and ${\sf Ord}(q_1) < {\sf Ord}(q_2)$,\\
			  	\noindent
			  	i.e. none has a predecessor in $Q_d$, then they are ordered using $\ord$;
				
				\item {\bf or}, $\exists q_1' \in t_1: q_1=\delta(q_1',\sigma)$, and $\neg\exists q'_2 \in t_1:q_2=\delta(q_2',\sigma)$,\\
				\noindent i.e. $q_1$ has a $\sigma$-predecessor in $Q_d$, and $q_2$ not;
				
				\item {\bf or} $\exists q'_1 \in t_1:q_1=\delta(q'_1,\sigma)$, and $\exists q'_2 \in t_1:q_2=\delta(q_2',\sigma)$, and $\min_{<_1} \{ q'_1 \in t_1 \mid q_1=\delta(q'_1,\sigma)\} < \min_{<_1} \{ q'_2 \in t_1 \mid q_2=\delta(q'_2,\sigma) \}$,\\
				\noindent
				i.e. both have a predecessor in $Q_d$, and they are ordered according to the order of their minimal parents.
			\end{enumerate}
		\end{itemize}
		
\paragraph{{\bf Coloring}} To define the coloring of edges in the deterministic automaton, we need to identify the states $q \in t_1$ in a transition $(s_1,(t_1,<_1)) \stackrel{\sigma}{\rightarrow} (s_2,(t_2,<_2))$ whose indices decrease when going from $t_1$ to $t_2$. Those are defined as follows:
		\[{\sf Dec}(t_1)= \{ q_1 \in t_1 \mid {\sf Ind}_{(t_2,<_2)}(\delta(q_1,\sigma)) < {\sf Ind}_{(t_1,<_1)}(q_1) \}.\]
\noindent
Additionally, let $\acc(t_1)=\{ q \mid \exists q' \in t_2 : (q,\sigma,q') \in  \alpha \}$ denote the subset of states in $t_1$ that are the source of an accepting transition.

We assign a color to each transition $(s_1,(t_1,<_1)) \rightarrow^{\sigma} (s_2,(t_2,<_2))$ as follows:

\begin{enumerate}
	\item if $\dec(t_1)=\emptyset$ and $\acc(t_1)\not=\emptyset$, the color is $2 \cdot\min_{q \in \acc(t_1)} \ind_{(t_1,<_1)}(q)$.
	\item if $\dec(t_1)\not=\emptyset$ and $\acc(t_1)=\emptyset$, the color is $2 \cdot\min_{q \in \dec(t_1)} \ind_{(t_1,<_1)}(q)-1$.
	\item if $\dec(t_1)\not=\emptyset$ and $\acc(t_1)\not=\emptyset$, the color is defined as the minimal color among 
		\begin{itemize}
			\item $c_{{\sf odd}}=2 \cdot\min_{q \in \dec(t_1)} \ind_{(t_1,<_1)}(q)-1$, and 
			\item $c_{{\sf even}}=2 \cdot \min_{q \in \acc(t_1)} \ind_{(t_1,<_1)}(q)$.
		\end{itemize}
	\item if $\dec(t_1)=\acc(t_1)=\emptyset$, the color is $2 \cdot |Q_q|+1$.
\end{enumerate}

\begin{figure}[t]
\myspace\myspace\myspace
	\begin{center}
		\begin{tikzpicture}
		[x=2cm,y=7mm,font=\small,initial text=,outer sep=0pt]
		
		\node[state,initial] (1) at (0,0) {$\{1\},[]$};
		\node[state] (2) at (1,1) {$\{1\},[2]$}; 
		\node[state] (3) at (1,-1) {$\{1\},[3]$};   	
		\node[state] (4) at (2,1) {$\{1\},[4< 3]$}; 
		\node[state] (5) at (2,-1) {$\{1\},[4< 2]$}; 
		
		\path[->]
		(1) edge[bend left] node[above]{a} node[below]{7}  (2)
		(2) edge[loop above] node[above]{a} node[below]{2} ()
		(1) edge[bend right] node[above]{7} node[below]{b}  (3)
		(3) edge[loop below] node[above]{2} node[below]{b} ()
		(2) edge node[above]{b} node[below]{4}  (4)
		(4) edge[loop above] node[above]{b} node[below]{4} ()
		(3) edge node[above]{4} node[below]{a}  (5)
		(5) edge[loop below] node[above]{4} node[below]{a} ()
		(4) edge[bend left] node[right]{a} node[left]{3}  (5)
		(5) edge[bend left] node[left]{b} node[right]{3}  (4)   	
		;
		\end{tikzpicture}
		~~~
		\begin{tikzpicture}
		[x=2cm,y=7mm,font=\small,initial text=,outer sep=0pt]
		
		\node[state,initial] (1) at (0,0) {$\{1\},[]$};
		\node[state] (2) at (1,1) {$\{1\},[2]$}; 
		\node[state] (3) at (1,-1) {$\{1\},[3]$};   	
		
		\path[->]
		(1) edge[bend left] node[above]{a} node[below]{3}  (2)
		(2) edge[loop above] node[above]{a} node[below]{2} ()
		(1) edge[bend right] node[above]{3} node[below]{b}  (3)
		(3) edge[loop below] node[above]{2} node[below]{b} ()
		(2) edge[bend left] node[right]{a} node[left]{1}  (3)
		(3) edge[bend left] node[left]{b} node[right]{1}  (2)   	
		;
		\end{tikzpicture}
	\end{center}
	\myspace\myspace
	\caption{\label{fig:DPA-ex} Left: DPA that accepts the  LTL language $\F \G a \lor \F \G b$, edges are decorated with a natural number that specifies its color. Right: A reduced DPA. }
	\myspace
\end{figure}
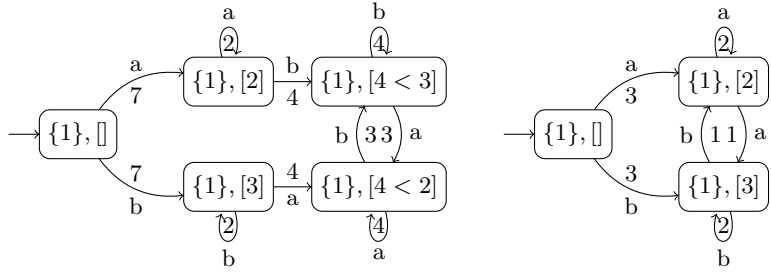

\begin{example}
The DPA of Fig.~\ref{fig:DPA-ex} is the automaton that is obtained by applying the construction LDBA$\rightarrow$DPA defined above to the LDBA of Fig.~\ref{fig:ex-LDBA} that recognizes the LTL language $\F \G a \lor \F \G b$. The figure only shows the reachable states of this construction. As specified in the construction above, states of DPA are labelled with a subset of $\overline{Q_d}$ and a ordered subset of $Q_d$ of the original NBA. As an illustration of the definitions above, let us  explain the color of edges from state $(\{1\},[4,3])$ to itself on letter $b$. When the NBA is in state $1$, $3$ or $4$ and letter $b$ is read, then the next state of the automaton is again  $1$, $3$ or $4$. Note also that there are no runs that are merging in that case. As a consequence, the color that is emitted is even and equal to the index of the smallest state that is the target of an accepting transition. In this case, this is state $3$ and its index is $2$. This is the justification for the color $4$ on the edge. On the other hand, if letter $a$ is read from state  $(\{1\},[4,3])$, then the automaton moves to states $(\{1\},[4,2])$. The state $3$ is mapped to state $4$ and there is a run merging which induces that the color emitted is odd and equal to $3$. This $3$ trumps all the $4$'s that were possibly emitted from state $(\{1\},[4,3])$ before.
\end{example}

  \begin{theorem}
  The language defined by the deterministic parity automaton $B$ is equal to the language defined by the limit deterministic automaton $A$, i.e. $\lang(A)=\lang(B)$.
  \end{theorem}
  \begin{proof}
  Let $w \in \Sigma^{\omega}$ and $G_w$ be the run DAG of $A$ on $w$. It is easy to show by induction that the sequence of 
  colors that occur along $G_w$ is equal to the sequence of colors defined by the run of the automaton $B$ on $w$. By Theorem~\ref{thm:coloring}, the language of automaton $B$ is thus equal to the language of automaton $A$. \qed
  \end{proof}
  

\subsection{Complexity Analysis}

\subsubsection{Upper bound}

Let $n = |Q|$ be the size of the LDBA and let $n_d = |Q_d|$ be the size of the accepting component. We can bound the number of different orderings using the series of reciprocals of factorials (with $e$ being Euler's number):

\[
|{\cal OP}(Q_d)| = \sum_{i=0}^{n_d}\frac{n_d!}{(n_d-i)!} \leq  n_d \cdot n_d!  \cdot \sum_{i=0}^{\infty}\frac{1}{i!} = e \cdot n_d \cdot n_d! \in\mathcal O(2^{n\cdot\log n})
\]

\noindent Thus the obtained DPA has ${\cal O}(2^n\cdot2^{n\cdot\log n}) = 2^{\mathcal O(n\cdot\log n)}$ states and ${\cal O}(n)$ colours.

\subsubsection{Lower bound}

We obtain a matching lower bound by strengthening Theorem~8 from \cite{DBLP:conf/fsttcs/Loding99}:


\begin{lemma}
There exists a family $(L_n)_{n \geq 2}$ of languages ($L_n$ over an alphabet of $n$ letters) such that for every $n$ the language $L_n$ can be recognized by a limit-deterministic Büchi automaton with $3n + 2$ states but can not be recognized by a deterministic Parity automaton with less than $n!$ states.
\end{lemma}

\begin{proof}
The proof of Theorem 8 from \cite{DBLP:conf/fsttcs/Loding99} constructs a non-deterministic Büchi automaton of exactly this size and which is in fact limit-deterministic. 

Assume there exists a deterministic Parity automata for $L_n$ with $m < n!$ states. Since parity automata are closed under complementation, we can obtain a parity automaton and hence also a Rabin automaton of size $m$ for $\overline{L_n}$ and thus a Streett automaton of size $m$ for $L_n$, a contradiction to Theorem~8 of \cite{DBLP:conf/fsttcs/Loding99}.\qed
\end{proof}

\begin{corollary}
Every translation from limit-deterministic Büchi automata of size $n$ to deterministic parity yields automata with $2^{\Omega(n \log n)}$ states in the worst case.
\end{corollary}

\section{From LTL to Parity in $2^{2^{\mathcal O(n)}}$}

In \cite{DBLP:conf/cav/SickertEJK16} we present a LTL$\rightarrow$LDBA translation. Given a formula 
$\varphi$ of size $n$, the translation produces an asymptotically optimal LDBA with $2^{2^{\mathcal O(n)}}$ states. 
The straightforward composition of this translation with the single exponential 
LDBA$\rightarrow$DPA translation of the previous section is only guaranteed to be 
triple exponential, while the Safra-Piterman construction produces a DPA of at most doubly 
exponential size. In this section we describe a modified composition that yields a double 
exponential DPA.
To the best of our knowledge this is is the first translation of the whole LTL to deterministic parity automata that is asymptotically optimal and does not use Safra's construction.

The section is divided into two parts. 
In the first part, we explain and illustrate a redundancy occurring in our LDBA$\rightarrow$DPA translation, responsible for the undesired extra exponential.
We also describe an optimization that removes this redundancy when the LDBA satisfies some conditions.
In the second part, we show these conditions are satisfied on the products of the LTL$\rightarrow$LDBA translation, which in turn guarantees a doubly exponential LTL$\rightarrow$DPA procedure. 

\subsection{An improved construction}

We can view the second component of a state of the DPA as a sequence of states of the LDBA, ordered by their indices.
Since there are $2^{2^{\mathcal O(n)}}$ states of the LDBA for an LTL formula of length $n$, the number of such sequences is 
$$2^{2^{\mathcal O(n)}}!=2^{2^{2^{\mathcal O(n)}}}$$
If only the length of the sequences (the maximum index) were bounded by $2^n$, the number of such sequences would be smaller than the number of functions $2^n\to 2^{2^{\mathcal O(n)}}$ which is
$$(2^{2^{\mathcal O(n)}})^{2^n}=2^{2^{\mathcal O(n)}\cdot 2^n}=2^{2^{\mathcal O(n)}}$$

Fix an LDBA with set of states $Q$. Assume the existence of an {\em oracle}: a list of statements of the form $\lang(q) \subseteq \bigcup_{q' \in Q_q} \lang(q')$ where $q \in Q$ and $Q_q \subseteq Q$.  We use the oracle to define a mapping that associates to each run DAG $G_w$ a ``reduced DAG'' $G_w^*$, defined as the result of iteratively performing the following four-step operation:
\begin{itemize}
\item Find the first $V_i$ in the current DAG such that the sequence $(v_1,i)\sqsubset(v_2,i)\sqsubset\cdots\sqsubset(v_{n_i},i)$ of vertices of 
$V_i^d$ contains a vertex $(v_k,i)$ for which the oracle ensures
\begin{equation}
\lang(v_k)\subseteq\bigcup_{j<k}\lang(v_j) \tag{$*$}\label{eq:lang-based}
\end{equation}
\noindent We call $(v_k, i)$ a {\em redundant vertex}.
\item Remove $(v_k, i)$ from the sequence, and otherwise keep the ordering $\sqsubseteq_i$ unchanged (thus decreasing the index of vertices $(v,\ell)$ with $\ell>k$). 
\item Redirect transitions leading  from vertices in $V_{i-1}$ to $(v_k, i)$ so that they lead to the smallest vertex $(v_1, i)$ of $V_i$. 
\item Remove any vertices (if any) that are no longer reachable from vertices of $V_1$. 
\end{itemize}
We define the color summary of $G_w^*$ in exactly the same way as the color summary of $G_w$.
The DAG $G_w^*$ satisfies the following crucial property, whose proof can be found in Appendix~\ref{app:merging}:
\begin{proposition}\label{prop:merging}
	The color summary of the run DAG $G_w^*$ is even if and only if there is an accepting run in $G_w$.
\end{proposition}

The mapping on DAGs induces a reduced DPA as follows. The states are the 
pairs $(s, (t, <))$ such that $(t, <)$ does not contain redundant vertices.
There is a transition $(s_1, (t_1, <)) \stackrel{a}{\rightarrow} 
(s_2, (t_2, <))$ with color $c$ if{}f there is a word $w$ and an index $i$ such that  $(s_1, (t_1, <))$ and $(s_2, (t_2, <))$ correspond to the $i$-th and $(i+1)$-th levels of $G_w^*$, and $a$ and $c$ are the letter and color of the step between these levels in $G_w^*$. Observe that the set of transitions is independent of the words chosen to define them.

The equivalence between the initial DPA $\A$ and the reduced DPA $\A_r$ follows immediately from Proposition \ref{prop:merging}: $\A$ accepts $w$ if{}f $G_w$ contains an accepting run if{}f the color summary of $G_w^*$ is even if{}f $\A_r$ accepts $w$.

\begin{example}
	Consider the LDBA of Fig.~\ref{fig:ex-LDBA} and an oracle given by $\lang(4)=\emptyset$, ensuring $\lang(4)\subseteq\bigcup_{i\in I}\lang(i)$ for any $I\subseteq Q$. Then $4$ is always redundant and merged, removing the two rightmost states of the DPA of Fig.~\ref{fig:DPA-ex} (left), resulting in the DPA of Fig.~\ref{fig:DPA-ex} (right). 
	However, for the sake of technical convenience, we shall refrain from removing a redundant vertex when it is the smallest one (with index $1$).
\end{example}

Since the construction of the reduced DPA is parametrized by an oracle,
the obvious question is how to obtain an oracle that does not involve applying an expensive language inclusion test. Let us give a first example in which an oracle can be easily obtained: 

\begin{example}
	Consider an LDBA where each state $v=\{s_1,\ldots,s_k\}$ arose from some powerset construction on an NBA in such a way that $\lang(\{s_1,\ldots,s_k\})=\lang(s_1)\cup\cdots\lang(s_k)$. An oracle can, for instance, allow us to merge whenever $v_k\subseteq\bigcup_{j<k}v_j$, which is a sound syntactic approximation of language inclusion. 
	This motivates the following formal generalization.
\end{example}

Let $\mathcal L_B=\{L_i\mid i\in B\}$ be a finite set of languages, called \emph{base} languages.
We call $\mathcal L_C:=\{\bigcup\mathcal L\mid\mathcal L\subseteq\mathcal L_B\}$ the join-semilattice of \emph{composed} languages.
We shall assume an LDBA with some $\mathcal L_B$ such that $\lang(q)\in\mathcal L_C$ for every state $q$.
We say that such an LDBA \emph{has a base} $\mathcal L_B$.
In other words, every state recognizes a union of some base languages.
(Note that every automaton has a base of at most linear size.)
Whenever we have states $v_j$ recognizing $\bigcup_{i\in I_j}L_i$ with $I_j\subseteq B$ for every $j$, the oracle allows us to merge vertices $v_k$ satisfying $I_k\subseteq \bigcup_{j<k}I_j$.
Intuitively, the oracle declares a vertex redundant  whenever the simple syntactic check on the indices allows for that.

Let $V_1=\bigcup_{i\in I_1}L_i,\cdots V_j=\bigcup_{i\in I_j}L_i$ be a sequence of languages of $\mathcal L_C$ where the reduction has been applied and there are no more redundant vertices. The maximum length of such a sequence
is given already by the base $\mathcal L_B$ and we denote it $\width(\mathcal L_B)$.
\begin{lemma}
	For any $\mathcal L_B$, we have $\width(\mathcal L_B)\leq|\mathcal L_B|+1$.
\end{lemma}
\begin{proof}
	We provide an injective mapping of languages in the sequence (except for $V_1$) into $B$.	
	Since $I_2\not\subseteq I_1$, there is some $i\in I_2\setminus I_1$ and we map $V_2$ to this $i$.
	In general, since $I_k\not\subseteq \bigcup_{j=1}^{k-1}I_j$, we also have $i\in I_k\setminus \bigcup_{j=1}^{k-1}I_j$ and we map $V_k$ to this $i$.\qed
\end{proof}

On the one hand, the transformation of LDBA to DPA without the reduction yields $2^{\mathcal O(|Q|\cdot\log|Q|)}$ states.
On the other hand, we can now show that the second component of reduced LDBA with a base can be exponentially smaller. Further, let us assume the LDBA is \emph{initial-deterministic}, meaning that $\delta\cap(\overline{Q_d}\times\Sigma\times\overline{Q_d})$ is deterministic, thus not resulting in blowup in the first component.

\begin{corollary}\label{cor:base}
	For every initial-deterministic LDBA with base of size $m$, there is an equivalent DPA with $2^{\mathcal O(m^2)}$ states.
\end{corollary}
\begin{proof}
	The number of composed languages is $\mathcal L_C=2^{m}$.
	Therefore, the LDBA has at most $2^m$ (non-equivalent) states.
	Hence the construction produces at most
	$$|\mathcal L_C|\cdot|\mathcal L_C|^{\mathcal O(\width(\mathcal L_B))}=2^m\cdot(2^m)^{\mathcal O(m)}=2^{\mathcal O(m^2)}$$
	states since the LDBA is initial-deterministic, causing no blowup in the first component. \qed
\end{proof}

\subsection{Bases for LDBAs Obtained from LTL Formulas}

We prove that the width for LDBA arising from the LTL transformation is only singly exponential in the formula size.
To this end, we need to recall a property of the LTL$\rightarrow$LDBA translation of \cite{DBLP:conf/cav/SickertEJK16}.
Since partial evaluation of formulas plays a major role in the translation, we introduce the following definition.
Given an LTL formula $\varphi$ and sets $T$ and $F$ of LTL formulas, let $\varphi[T,F]$ denote the result of substituting $\true$ (true) for each occurrence of a formula of $T$ in $\varphi$, and similarly $\false$ (false) for formulas of $F$.
The following property of the translation is proven in Appendix~\ref{app:ltl}.

\begin{proposition}\label{prop:translation}
	For every LTL formula $\varphi$, every state $s$ of the LDBA of \cite{DBLP:conf/cav/SickertEJK16} is labelled by an LTL formula $\mathit{label}(s)$ such that (i) $\lang(s)=\lang(\mathit{label}(s))$ and (ii) $\mathit{label}(s)$ is a Boolean combination of subformulas of $\varphi[T_s, F_s]$ for some $T_s$ and $F_s$.
	Moreover, the LDBA is initial-deterministic.
\end{proposition}

 As a consequence, we can bound the corresponding base:

\begin{corollary}\label{cor:ltl}
	For every LTL formula $\varphi$, the LDBA of \cite{DBLP:conf/cav/SickertEJK16} for $\varphi$ has a base of size $2^{\mathcal O{(|\varphi|)}}$.
\end{corollary}
\begin{proof}
Firstly, we focus on states using the same $\varphi[T_s, F_s]$.
The language of each state can be defined by a Boolean formula over $\mathcal O (|\varphi|)$ atoms.
Since every Boolean formula can be expressed in the disjunctive normal form, its language is a union of the conjuncts.
The conjunctions thus form a base for these states.
There are exponentially many different conjunction in the number of atoms.
Hence the base is of singly exponential size $2^{\mathcal O(|\varphi|)}$ as well.

Secondly, observe that there are only $2^{\mathcal O(|\varphi|)}$ different formulas $\varphi[T_s, F_s]$ and thus only $2^{\mathcal O(|\varphi|)}$ different sets of atoms.	
Altogether, the size is bounded by 
\begin{equation*}
2^{\mathcal O(|\varphi|)}\cdot 2^{\mathcal O (|\varphi|)}= 2^{\mathcal O (|\varphi|)} \qedhere
\end{equation*}
\end{proof}

\begin{theorem}
	For every LTL formula $\varphi$, there is a DPA with $2^{2^{\mathcal O(|\varphi|)}}$ states.
\end{theorem}
\begin{proof}
The LDBA for $\varphi$ has base of singly exponential size $2^{\mathcal O(|\varphi|)}$ by Corollary~\ref{cor:ltl} and is initial-deterministic by Proposition~\ref{prop:translation}. Therefore, by Corollary \ref{cor:base}, the size of the DPA is doubly exponential, in fact
\begin{equation*}
2^{{(2^{\mathcal O(|\varphi|)})}^2}=2^{2^{\mathcal O(|\varphi|)}} \qedhere
\end{equation*}
\end{proof}
This matches the lower bound $2^{2^{\Omega(n)}}$ by \cite{DBLP:conf/mochart/KupfermanR10} as well as the upper bound by the Safra-Piterman approach. Finally, note that while the breakpoint constructions in \cite{DBLP:conf/cav/SickertEJK16} is analogous to Safra's vertical merging, the merging introduced here is analogous to Safra's horizontal merging.

\newcommand{\veeton}{\ensuremath{\bigvee_{i=1}^n}}
\newcommand{\wedgeton}{\ensuremath{\bigwedge_{i=1}^n}}

\section{Experimental Evaluation}

We evaluate the performance of our construction on several datasets taken from 
\cite{DBLP:conf/lpar/BlahoudekKS13,DeWulf:2008vt,DBLP:conf/cav/SickertEJK16} and several Temporal 
Logic Synthesis Format (TLSF) specifications \cite{DBLP:journals/corr/JacobsBBK0KKLNP16} of
the SyntComp 2016 competition. 

We use the size of the constructed deterministic automaton as an indicator 
for the overall performance of the synthesis procedure. In \cite{DBLP:conf/charme/SebastianiT03} 
it is argued that the degree of determinism of the automaton is a better predictor for performance
in model-checking problems; however, this parameter is not applicable for synthesis problems, which
require deterministic automata. 

We compare two versions of our implementation (with and without optimizations, see below)
with the algorithms of Spot \cite{duret.16.atva2}. Each tool is given 64GB
of memory and 10 minutes. Increasing time to 10 hours
does not change the results. More precisely, we compare the following three setups:

\medskip \noindent \textbf{S.} (\texttt{ltl2tgba}, 2.1.1) - Spot \cite{duret.16.atva2} implements a
version of the Safra-Piterman determinization procedure \cite{DBLP:journals/fuin/Redziejowski12}
with several optimizations.

\medskip \noindent \textbf{L2P and L2P$'$.}  (\texttt{ltl2dpa}, 1.0.0) - L2P is the construction of this paper,
available at \url{www7.in.tum.de/~sickert/projects/ltl2dpa}. 
L2P$'$ adds two optimizations. First, the tool translates both the formula and its negation to 
DPAs $A_1, A_2$, complements $A_2$ to yield $\overline{A}_2$, and picks the smaller of $A_1, A_2$. 
Further, we apply the simplification routines of Spot (\texttt{ltlfilt} and \texttt{autfilt}, respectively). 

\begin{figure}[t]
	\vspace*{-2cm}

%
%
\subfloat[Parametrised Formulas]{\begin{tikzpicture}[scale=0.69]
\begin{loglogaxis}[xlabel={L2P$_{\parallel}^*$ (states)}, ylabel={S (states)}, xmax=10000, ymax=10000, ymin=0, xmin=0, enlargelimits=false, legend style={at={(0.98,0.02)},anchor=south east}]
\addplot[mark=diamond, smooth] table[x=L2D-PAR++,y=S]{GR1.dat};
\addplot[mark=square, smooth] table[x=L2D-PAR++,y=S]{R.dat};
\addplot[mark=triangle, smooth] table[x=L2D-PAR++,y=S]{F.dat};
\addplot[mark=o, smooth] table[x=L2D-PAR++,y=S]{theta.dat};
\addplot[mark=x] table[x=L2D-PAR++,y=S]{C1.dat};
\addplot[mark=x] table[x=L2D-PAR++,y=S]{C2.dat};
\addplot[mark=x] table[x=L2D-PAR++,y=S]{E.dat};
\addplot[mark=x] table[x=L2D-PAR++,y=S]{Q.dat};
\addplot[mark=x] table[x=L2D-PAR++,y=S]{S.dat};
\addplot[mark=x] table[x=L2D-PAR++,y=S]{U1.dat};
\addplot[mark=x] table[x=L2D-PAR++,y=S]{U2.dat};
\addplot[domain=1:10000, smooth]{x};
\legend{$G(n)$,$R(n)$,$F(n)$,$\theta(n)$,other}
\end{loglogaxis}
\end{tikzpicture}
\label{fig:param}}
%
%
\subfloat[Randomly Generated Formulas]{\begin{tikzpicture}[scale=0.69]
\begin{loglogaxis}[xlabel={L2P$_{\parallel}^*$ (states)}, ylabel={S (states)}, xmax=1000, ymax=1000, ymin=1, xmin=1, enlargelimits=false, legend style={at={(0.98,0.02)},anchor=south east}]
\addplot[only marks, mark=x] table[x=L2D-PAR++,y=S]{lpar_fg.dat};
\addplot[only marks, mark=x] table[x=L2D-PAR++,y=S]{lpar_morefg.dat};
\addplot[only marks, mark=x] table[x=L2D-PAR++,y=S]{lpar_uniform.dat};
\addplot[domain=1:1000, smooth]{x};
\end{loglogaxis}
\end{tikzpicture}
\label{fig:random}} \\
%
%
\subfloat[\enquote{Real-world} Formulas]{\begin{tikzpicture}[scale=0.69]
\begin{axis}[xlabel={L2P$_{\parallel}^*$ (states)}, ylabel={S (states)}, xmode=log, ymode=log, xmax=100, ymax=100, ymin=1, xmin=1, enlargelimits=false]
\addplot[only marks, mark=x] table[x=L2D-PAR++,y=S]{lpar_real.dat};
\addplot[domain=1:100, smooth]{x};
\end{axis}
\end{tikzpicture}
\label{fig:real}} 
%
%
\subfloat[TLSF/Acacia]{
\begin{tikzpicture}[scale=0.69]
\begin{axis}[xlabel={L2P$_{\parallel}^*$ (states)}, ylabel={S (states)}, xmode=log, ymode=log, xmax=1000, ymax=1000, ymin=1, xmin=1, enlargelimits=false]
\addplot[only marks, mark=x] table[x=L2D-PAR++,y=S]{syntcomp_acacia.dat};
\addplot[domain=1:10000, smooth]{x};
\end{axis}
\end{tikzpicture}
\label{fig:acacia}}
	\caption{Comparison of Spot and our implementation using the best configurations. Timeouts are denoted by setting the size of the automaton to the maximum.}
        \label{fig:scatter}
	
	\medskip
	
	\captionof{table}{Number of states and number of used colours in parenthesis for the constructed automata. Timeouts are marked with $t$.}
	\label{tab:exp}

\newcolumntype{Z}{>{\centering\let\newline\\\arraybackslash\hspace{0pt}}X}

\begin{center}
\scalebox{0.8}{
\begin{minipage}{1.2\textwidth}
\begin{tabularx}{\textwidth}{l|*{5}{Z}|*{4}{Z}|Z}
\toprule
       				& $f(1,0)$ 	& $f(1,2)$	& $f(1,4)$ & $f(2,0)$ & $f(2,2)$ & zn	  & zp1		& zp2	  & zp3	& Buffer \\
\cmidrule{1-11}				
S     				& 18(6)	& 141(8)	& 2062(8) & 208(12) & 883(12) & t 	  & t 	  	& t 	  & t 	& t 	      \\
L2P 				& 12(8)	& 114(9)	& 332(15) & 144(14) & 4732(19)& t 	  & t 	  	& t 	  & t 	& 1425(27)   \\
L2P$'$ & 12(8) & 78(7)		& 271(11) & 106(9)  & 1904(15)& 32(6) & 42(6)  & 111(12)& 97(12) & 435(4) \\
\bottomrule
\end{tabularx}
\end{minipage}
}
\end{center}

	\myspace\myspace
\end{figure}

\medskip

We consider three groups of benachmarks:

\medskip

\noindent \textbf{Parametric Formulas.} 
10 benchmarks from  \cite{DBLP:conf/lpar/BlahoudekKS13,DBLP:conf/cav/SickertEJK16}).
In six cases S and L2P$'$ produce identical results. The other four are
\[\begin{array}{rlrl}
R(n) & = \wedgeton (\G\F p_i \vee \F\G p_{i+1}) & G(n) & = (\wedgeton \G\F p_i) \rightarrow (\wedgeton \G\F q_i) \\
\theta(n) & = \neg ((\wedgeton \G \F p_i) \rightarrow \G(q \rightarrow \F r)) & F(n) & = \wedgeton (\G\F p_i \rightarrow \G\F q_i) \\
\end{array}\]
\noindent for which the results are shown in (figure \ref{fig:scatter}a). Additionally, we consider the 
``$f$'' formulas from \cite{DBLP:conf/cav/SickertEJK16} (table \ref{tab:exp}). 
Observe that L2P$'$ performs clearly better, and the gap between the tools grows when the parameter increases. 

\medskip \noindent \textbf{Randomly Generated Formulas} from \cite{DBLP:conf/lpar/BlahoudekKS13} (figure \ref{fig:scatter}b). 

\medskip \noindent \textbf{Real Data.} Formulas taken from case studies and synthesis competitions --- the intended domain of application of our approach. Figures \ref{fig:real} and \ref{fig:acacia} show results for 
the real-world formulas of \cite{DBLP:conf/lpar/BlahoudekKS13} and the TLSF specifications contained in the Acacia set of \cite{DBLP:journals/corr/JacobsBBK0KKLNP16}. Table \ref{tab:exp} shows results for LTL formulas expressing properties of
Szymanski's protocol \cite{DeWulf:2008vt}, and for the generalised buffer benchmark of Acacia.

\medskip \noindent \textbf{Average Compression Ratios.}  The geometric average compression ratio for a benchmark suite $B$ 
is defined as ${\prod_{\varphi \in B}  (n_\varphi^S / n^{L2P'}_\varphi)}^{1/|B|}$, 
where $n_\varphi^S$ and $n^{L2P'}_\varphi $ denote the number of states of the automata produced 
by Spot and L2P$'$, respectively. The ratios in our experiments (excluding benchmarks where Spot times out)
are: 1.14 for random formulas, 1.12 for the real-world formulas of \cite{DBLP:conf/lpar/BlahoudekKS13},
and 1.35 for the formulas of Acacia.
\newpage

\section{Conclusion}

We have presented a simple, \enquote{Safraless}, and asymptotically optimal  translation from LTL and LDBA to deterministic parity automata. 
Furthermore, the translation is suitable for an on-the-fly implementation. The resulting automata are substantially smaller
than those produced by the SPOT library for formulas obtained from synthesis specifications, and have comparable or smaller size
for other benchmarks. In future work we want to investigate the performance of the translation as part of a synthesis toolchain.

\paragraph{Acknowledgments.} The authors want to thank Michael Luttenberger for helpful discussions and the anonymous reviewers for constructive feedback.

\bibliographystyle{alpha}
\bibliography{ref}

\newpage
\appendix
\section{Proof of Theorem \ref{thm:coloring}}\label{app:coloring}

\begin{reftheorem}{thm:coloring}
The color summary of the run DAG $G_w$ is even if and only if there is an accepting run in $G_w$.
\end{reftheorem}

\begin{proof}
	``$\Rightarrow$'':
Assume that the color summary of $G_w$ is even and equal to $c$. Then it must be the case that there exists a level $i \geq 0$ such the color after level $i$ is always larger than or equal to $c$, and infinitely many times equal to $c$. W.l.o.g. assume that in level $i$, there exists a vertex $v=(q,i) \in \acc(V^d_i)$ and $c=2 \cdot \ind(v)$. Take the smallest run prefix that ends up in $v$, this run prefix will never merge with a smaller run prefix and all smaller run prefixes that are active in level $i$ will not merge, as otherwise, there would exist a position $j \geq i$ where the index of the run that passes by $(q,i)$ would decrease and this would contradict the fact that for all $j \geq i$, all the colors that are emitted are larger than or equal to $c$. Let us now consider the suffix of the run that pass by $v=(q,i)$. As the even color $c$ is emitted infinitely many times after level $i$, we know that this run suffix crosses infinitely many times $\alpha$. So this run is accepting and this is the smallest such run. 

``$\Leftarrow$'':

(Step 1):
Now, let us consider the other direction. Assume that there exists an accepting run of $A$ on a word $w$. We first establish the existence of a run $\rho$ which is accepting and for which there exists a position $k \geq 0$ from which $\rho$ does not merge with any smaller run, and all smaller runs are non accepting. We identify $\rho$ and $k$ as follows. Among the accepting runs, we select one that enters first in the set of states $Q_d$ say at level $i \geq 0$. They can be several of them but we take one that enters $Q_d$ via a state $q$ of minimal index for ${\sf Ord}$. Let $V_i^d$ be the active states at level $i$ that are in $Q_d$. The way we have chosen $q$ make sure that all the states in $V^d_i$ with a smaller index than $q$ are the origin of non accepting runs and clearly as $\rho$ is accepting it cannot merge with one of those smaller runs. Now, some of those smaller runs may merge in the future, and each time they merge, the index of $\rho$ will decrease. But this will happen a number of times which is bounded by $Q_d$. 

(Step 2): 
Let $k$ be the position when the last merge of a smaller run pefix happens.

(Step 3): 
Let us now show that the existence of $\rho$ and this position $k$ allow us to prove that the color summary is even. After position $k$, there are only odd colors with values larger than or equal to $2 \cdot \ind(\rho(k))+1$ because we know that nor $\rho$ neither smaller runs merge in the future. Also as $\rho$ is accepting, there will be an infinite number of positions $l \geq k$ where the even color is equal to $2 \cdot \ind(\rho(k)))$, and only finitely many positions after $k$ may have an even color which is less than this value as all runs that are smaller than $\rho$ are not accepting. So the summary color is even and equal to $2 \cdot \ind(\rho(k)))$.
\end{proof}

\newpage

\section{Proof of Proposition \ref{prop:merging}}\label{app:merging}

\begin{refproposition}{prop:merging}
	The color summary of the run DAG $G_w^*$ is even if and only if there is an accepting run in $G_w$.
\end{refproposition}
\begin{proof}
	``$\Rightarrow$'':
	The `only-if' direction can be proven as in Theorem~\ref{thm:coloring} verbatim, only replacing $G_w$ by $G_w^*$.
	The reason why the argumentation is still correct, is that the discussed ``smallest run prefix that ends up in $v$'' (now in $G_w^*$) is actually a real run prefix (in $G_w$) since it never secondarily merged. Indeed, runs only merge into smaller ones.
	
	``$\Leftarrow$'':
	(Step 1):
	For the `if' direction, we first use the proof Theorem~\ref{thm:coloring}, Step 1, verbatim, obtaining the smallest accepting run in $G_w$.
	
	Additionally, we prove that this (smallest) constructed run $\rho$ is actually a run in $G_w^*$. 
	For a contradiction, assume that this is not the case and $\rho=\rho_1 (v_k,i) \rho_2$ where $(v_k,i)$ is the first vertex on $\rho$ that secondarily merged.
	Then there is $(v_j,i)\in V_i\cap(Q_d\times\{i\})$ with $v_j\sqsubset v_k$ and $\lang(v_j)$ contains the label of the run $(v_k,i) \rho_2$, accepted by some run $(v_j,i)\rho_2'$ in $G_w$.
	Since $(v_j,i)\sqsubset_i(v_k,i)$, we also have a run prefix $\rho_1'(v_j,i)\sqsubset\rho_1(v_k,i)$, and thus an accepting run $\rho_1'(v_j,i)\rho_2'$ in $G_w$ such that $\rho_1'(v_j,i)\rho_2'\sqsubset \rho_1 (v_k,i) \rho_2=\rho$, a contradiction with minimality of $\rho$. 
	
	(Step 2): Let $k$ be the position when the last merge of a smaller run pefix happens in $G_w^*$ (not $G_w$).
	
	(Step 3): We use the proof Theorem~\ref{thm:coloring}, Step 3, verbatim, proving the color summary is even.	
	\qed
\end{proof}

\section{Proof of Proposition \ref{prop:translation}}\label{app:ltl}

We start by recalling the LTL$\rightarrow$LDBA translation of \cite{DBLP:conf/cav/SickertEJK16}.

\paragraph{Preliminaries.} 
The translation assumes that formulas are in {\em negation normal form}, given 
by the syntax
\begin{align*}
\varphi::= & \; \true \mid \false \mid a \mid \neg a \mid \varphi\wedge\varphi \mid \varphi\vee\varphi \mid \X\varphi \mid \F\varphi \mid \G\varphi \mid \varphi\U\varphi
\end{align*}
\noindent where $a$ belongs to a finite set of atomic propositions. Every formula over the usual 
syntax of LTL (with negation and the $\X$ and $\U$ operators) can be normalized with linear blowup 
if formulas are represented by their syntax DAGs, where two occurrences of the same subformula are
represented by the same node.

We recall the ${\it af}$ function introduced in \cite{DBLP:conf/cav/EsparzaK14,DBLP:conf/cav/SickertEJK16},
and some of its properties. Let $\nu$ be a letter. The formula 
$\aft(\varphi,\nu)$, read ``$\varphi$ after $\nu$'' is  inductively defined as follows \cite{DBLP:conf/cav/EsparzaK14,DBLP:conf/cav/SickertEJK16}:
\[\begin{array}[t]{lcl}
\aft(\true,\nu)& = & \true \\[0.1cm]
\aft(\false,\nu)& = & \false \\[0.1cm]
\aft(a,\nu)& = & \left\{
\begin{array}{ll}
\true & \mbox{if $a \in \nu$} \\[0.1cm]
\false & \mbox{if $a \notin \nu$}\end{array}\right. \\[0.3cm]
\aft(\neg a,\nu)&=& \left\{
\begin{array}{ll}
\false & \mbox{if $a \in \nu$} \\ [0.1cm]
\true & \mbox{if $a \notin \nu$}\end{array}\right. \\[0.3cm]
\end{array} 
\hspace{0.5cm}
\begin{array}[t]{lcl}
\aft(\varphi\wedge\psi,\nu)&=& \aft(\varphi,\nu)\wedge\aft(\psi,\nu) \\[0.1cm]
\aft(\varphi\vee\psi,\nu)&=& \aft(\varphi,\nu)\vee\aft(\psi,\nu) \\[0.1cm]
\aft(\X\varphi,\nu)& = &  \varphi\\[0.1cm]
\aft(\G\varphi,\nu)&= & \aft(\varphi,\nu)\wedge \G\varphi\\[0.1cm]
\aft(\F\varphi,\nu)&= & \aft(\varphi,\nu)\vee \F\varphi\\[0.1cm]
\aft(\varphi\U\psi,\nu)&=& \aft(\psi,\nu)\vee(\aft(\varphi,\nu)\wedge \varphi\U\psi)
\end{array}\]
Furthermore, we define: $\aft(\varphi, \epsilon) = \varphi$, and
$\aft(\varphi, \nu w) = \aft(\aft(\varphi,\nu), w)$ for every letter $\nu$ and every finite word $w$.
The function $\aft$ has the following two properties \cite{DBLP:conf/cav/EsparzaK14,DBLP:conf/cav/SickertEJK16}
for every formula $\varphi$, finite word $v$, and $\omega$-word $w$:
\begin{itemize}
	\item[(i)]  $vw \models \varphi$ if{}f $w \models \aft(\varphi,v)$.
	\item[(ii)]  $\aft(\varphi, v)$ is a boolean combination of subformulas of $\varphi$.
\end{itemize}

A formula is {\em proper} if it is neither a conjunction nor a disjunction.
The propositional formula $\varphi_P$ of a formula $\varphi$ is the result of
substituting every maximal proper subformula $\psi$ of $\varphi$ by a propositional variable $x_\psi$. 
For example, if $\varphi = \X b \vee (\G(a \vee \X b) \wedge \X b)$ with $\psi_1 = \X b$ and $\psi_2 = \G(a \vee \X b)$, 
then $\varphi_P = x_{\psi_1} \vee (x_{\psi_2} \wedge x_{\psi_1})$. 
Two formulas $\varphi, \varphi'$ are {\em propositionally equivalent}, denoted $\varphi \equiv_P \varphi'$, 
if $\varphi_P$ and $\varphi'_P$ are equivalent. 
So, for example, $\X b$ is propositionally equivalent to $\X b \vee (\G(a \vee \X b) \wedge \X b)$.
Observe that propositional equivalence implies equivalence, but the contrary does not hold. For example,
$\F a \wedge \G a$ and $\G a$ are equivalent, but not propositionally equivalent. 

The states of the LDBA for an LTL formula are equivalence classes of formulas (or tuples thereof) 
with respect to propositional equivalence. However, we abuse language and write that the 
states are formulas or tuples of formulas.

\paragraph{Translating LTL to LDBA.}

\newcommand{\deltain}{\delta_{\it in}}
\newcommand{\Ain}{\A_{\it in}}
\newcommand{\Aac}{\A_{\it ac}}

Fix a formula $\varphi$.  We describe the LDBA $\A_\varphi$. We use 
$\varphi = c \vee\X \G (a \vee \F b)$ as running example. We abbreviate
$\psi := (a \vee \F b)$, and write $\varphi = c \vee \X\G\psi$. 
The LDBA $\A_\varphi$ is shown in Figure \ref{fig:complete}.

\begin{figure}[ht]
	\begin{center}
		\begin{tikzpicture}[x=2cm,y=1.5cm,font=\footnotesize,initial text=,outer sep=\os]
		
		\node[state,initial] (a) at (0,0) {$\varphi$};
		\node[state] (b) at (0,-1) {$\G\psi$};
		\node[state] (c) at (0,-2) {$\G\psi \wedge \F b$};
		
		\node[state] (d) at (1.7,0) {$\true$};
		\node[state] (1c) at (1.7,-1) {$\langle c, \cdot \rangle$};
		\draw node at (3.4,-1) {\textcolor{red}{$\Large c$}};
		\node[state] (1t) at (1.7,-2) {$\langle \true, \cdot \rangle$};
		\draw node at (3.4,-2) {\textcolor{red}{$\Large \true$}};
		
		\node[state] (1F) at (0,-3) {$\langle \F b, (\psi, \true) \rangle$};
		\draw node at (-1.6,-3) {\textcolor{red}{$\Large \G\psi \wedge \F b$}};
		\node[state] (2F) at (0,-4) {$\langle \F b, (\F b, \psi) \rangle$};
		\draw node at (-1.6,-4) {\textcolor{red}{$\Large \G\psi \wedge \F b$}};
		\node[state] (3F) at (0,-5) {$\langle \F b, (\F b, \F b) \rangle$};
		\draw node at (-1.6,-5) {\textcolor{red}{$\Large \G\psi \wedge \F b$}};
		
		\node[state] (1) at (1.7,-3)      {$\langle \true, (\psi, \true) \rangle$};
		\draw node at (3.4,-3) {\textcolor{red}{$\Large \G\psi \wedge \psi$}};
		\node[state] (2) at (1.7,-4)     {$\langle \true, (\F b, \psi) \rangle$};
		\draw node at (3.4,-4) {\textcolor{red}{$\Large \G\psi \wedge \F b$}};
		\node[state] (3) at (1.7,-5)     {$\langle \true, (\F b, \F b) \rangle$};
		\draw node at (3.4,-5) {\textcolor{red}{$\Large \G\psi \wedge \F b$}};
		
		\draw[dashed] (-1.0,-2.35) -- (0.75, -2.35) -- (0.75, -0.60) -- (3.3, -0.60);
		
		\path[->]
		(a) edge node[left] {$\bar{c}$} (b)
		edge node[above] {$c$} (d)
		edge node[above] {$\epsilon$} (1c)
		edge node[above] {$\epsilon$} (1)
		(b) edge [loop right, max distance=6mm,in=165,out=195,looseness=15] node[left] {$a+b$} (b)
		edge [bend right=-30] node[left]{$\bar{a}\bar{b}$} (c)
		edge node[above]{$\epsilon$} (1)
		(c) edge [bend left=30]  node[left] {$b$} (b)
		edge node[right]{$\epsilon$} (1F)
		edge [loop left,max distance=6mm,in=165,out=195,looseness=15] node[left]{$\bar{b}$} (c)
		(d)  edge [loop left, max distance=6mm,in=20,out=340,looseness=15] node[right] {$\true$} (d)
		edge [bend left=60] node[right] {$\epsilon$} (1t)
		edge [bend left=75] node[right] {$\epsilon$} (1);

		\path[->] (1F) edge node[right]{$\bar{a}\bar{b}$} (2F)
		edge node[below]{$b$} (1)
		edge [loop left, max distance=6mm,in=165,out=195,looseness=15] node[left]{$a\bar{b}$} (1F)
		(2F) edge node[right] {$\bar{a}\bar{b}$} (3F)
		edge node[below] {$b$} (1)
		edge  [loop left, max distance=6mm,in=165,out=195,looseness=15] node[left]{$a\bar{b}$} (2F)
		(3F) edge [loop below, max distance=6mm,in=250,out=290,looseness=15] node[below]{$\bar{b}$} (3F)
		edge node[below]{$b$} (1);
		
		\path[->] (1) edge node[right]{$\bar{a}\bar{b}$} (2)
		edge [loop right, max distance=9mm,in=60,out=90,looseness=5, double] node[right]{$a + b$} (1)
		(2) edge node[right] {$\bar{a}\bar{b}$} (3)
		edge[bend left=-60, double] node[right] {$b$} (1)
		edge  [loop right, max distance=4mm,in=340,out=20,looseness=15]node[right]{$a\bar{b}$} (2)
		(3) edge [loop below, max distance=6mm,in=250,out=290,looseness=15] node[below]{$\bar{b}$} (3)
		edge [bend left=-90, double] node[right]{$b$} (1);
		
		\path[->] (1c) edge node[right]{$c$} (1t)
		(1t) edge [loop right, max distance=6mm,in=340,out=20,looseness=15, double] node[right] {$\true$} (1t);
		
		\end{tikzpicture}
	\end{center}
	\caption{Automaton $\compaut$ for $\varphi = c \vee \X \G (a \vee \F b)$. The initial component is above the dashed line, the accepting component below.}
	\label{fig:complete}
\end{figure}
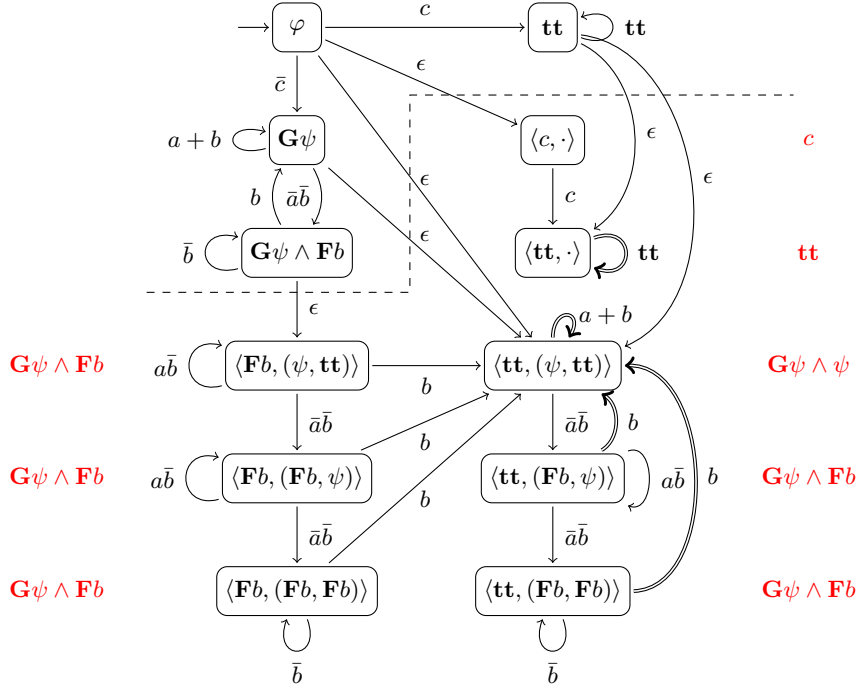

The LDBA $\A_\varphi$ consists of two deterministic components, called the {\em initial} and {\em accepting} 
components, and denoted $\Ain$ and $\Aac$, respectively---in Figure \ref{fig:complete} they 
are shown above and below the dashed line. The accepting component $\Aac$ is the union 
(defined componentwise for states, transitions, and accepting states) of subcomponents 
$\A_\setG$, one for each set $\setG$ of $\G$-subformulas of $\varphi$--that is, if $\varphi$
has $n$ different $\G$-subformulas, then $\Aac$ is the union of $2^n$ subcomponents). 
Transitions of $\A_\varphi$ labeled by an alphabet letter connect either two states 
of $\Ain$, or two states of the same subcomponent of $\Aac$. 
Further, for each state $q$ of $\Ain$ and each set $\setG$ there is an $\epsilon$-transition 
leading from $q$ to a state of $\A_\setG$. 

\medskip

\noindent {\bf Initial component $\Ain$:} Define the set of formulas {\em reachable} from $\varphi$ as 
$\Reach(\varphi)~=~\{\psi \mid \exists w. ~ \psi = \aft(\varphi,w)\}$. The set of states of $\Ain$  is ${\it Reach}(\varphi)$. 
The initial state is $\varphi$. The transition function $\deltain$ is given by 
$\deltain(\psi, a) = \aft(\psi, a)$. 
Intuitively, $\Ain$ monitors the formula that has to hold at the current moment for
$\varphi$ to hold at the beginning.

\medskip

\noindent {\bf Accepting component $\Aac$:} The accepting component $\Aac$ 
is the union of subcomponents $\A_\setG$, one for each $\setG \subseteq \gsf(\varphi)$.

Let $\gsf(\varphi)$ denote the set of all $\G$-subformulas of $\varphi$. 
Given a set $\setG \subseteq \gsf(\varphi)$ and a formula $\psi$, we write $\psi[\setG]$ as an abbreviation
for $\psi[\setG, \gsf(\varphi) \setminus \setG]$, i.e., for the 
result of substituting $\true$ for each maximal occurrence of a formula of $\setG$ in $\psi$,
and $\false$ for each maximal occurrence of a formula of $\gsf(\varphi) \setminus \setG$ in $\psi$.
For example, if $\setG = \{\G(a \vee \G b)\}$ then $\G b \vee \X (a \wedge \G(a \vee \G b))[\setG] = \false \vee 
\X (a \wedge \true) \equiv \X a$. 

Each subcomponent $\A_\setG$ is a product of DBAs: One for the formula $\varphi[\setG]$, and one for each  
formula of the form $\G(\psi[\setG])$, where $\G\psi \in \setG$. Observe that $\varphi[\setG]$ is a $\G$-free formula,
and $\G(\psi[\setG])$ does not have nested $\G$s. For example, if $\setG = \{ \G (a \vee \G b)\}$, then $\A_\setG$ is the
 product of three DBAs, one for $\varphi[\setG]$, one for $\G (b [\setG]) = \G b$, and a third one for $\G ((a \vee \G b)[\setG]) = \G (a \vee \false) \equiv_P \G a$. We call the DBAs for $\varphi[\setG]$ and $\G(\psi[\setG])$ the 
{\em monitors}. 

\medskip

\noindent {\bf Monitor for $\varphi[\setG]$.} The set of states is ${\it Reach}(\varphi[\setG])$, the transition function $\delta_{\varphi[\setG]}$ is given by $\delta_{\varphi[\setG]}(\psi, a) = \aft(\psi, a)$. The only final state is $\true$.
The initial state is left unspecified. 

\begin{itemize}
\item[$\star$] Lemma 2 of \cite{DBLP:conf/cav/SickertEJK16} shows that the 
$\varphi[\setG]$-monitor accepts a word $w$ from a state $q$ 
if{}f $w$ satisfies the formula $q$.
\end{itemize}

\noindent {\bf Monitor for $\G (\psi[\setG]$).} Let us abbreviate $\psi[\setG]$ as $\psi'$. The monitor for $\G\psi'$
is the DBA $\oracledet(\G\psi') = (2^{Ap}, \Reach(\psi') \times \Reach(\psi'), \delta, (\psi', \true), F)$ where
\begin{itemize}
	\item $\delta(\, (\xi_1, \xi_2) \, , \nu) = \begin{cases} (\, \aft(\xi_2, \nu) \wedge \psi' \, , \, \true \, ) & \text{if } \aft(\xi_1, \nu) \equiv_P \true \\[0.1cm] (\, \aft(\xi_1, \nu) \, , \, \aft(\xi_2, \nu) \wedge \psi' \, ) & \text{otherwise} \end{cases}$ \\[0.1cm]
	\item $F = \{(\, (\xi_1, \xi_2) \, , \nu, p) \in Q \times 2^{Ap} \times Q \mid \aft(\xi_1, \nu) \equiv_P \true\}$
\end{itemize}

\begin{itemize}
\item[$\star$] Lemma 5 of \cite{DBLP:conf/cav/SickertEJK16} proves that $\oracledet(\G\psi')$ accepts a word $w$ (from its initial state $(\psi', \true)$) if{}f $w \models \G\psi'$. 
\end{itemize}

\noindent {\bf Product.} Fix a set $\setG = \{ \G\psi_1, \ldots, \G\psi_n\} \subseteq \gsf$ of $\G$-subformulas
of $\varphi$. For every index $1 \leq i \leq n$, let $\oracledet_i = (2^{\it Ap},Q_i, \delta_i, q_{0i}, F_i)$ be the 
monitor for $\G(\psi_i[\setG])$. The {\em product} of these monitors is the generalized deterministic B\"uchi automaton 
\[
\prodaut(\setG) = 
(\; 2^{Ap} , 
 \; \prod_{i=1}^n Q_i , 
 \; \prod_{i=1}^n \delta_i , 
 \; (q_{01}, \ldots, q_{0n}) , 
 \; \{F_1', \ldots, F_n'\} \; )
\] 
where $( \; (q_1, \ldots, q_n), \nu, (q, q_1', \ldots, q_n') \; )$ is  a transition of 
$F_i'$ if{}f $(q_i, \nu, q_i') \in F_i$.

\begin{itemize}
\item[$\star$] Lemma 5 of \cite{DBLP:conf/cav/SickertEJK16} proves that $\prodaut(\setG)$ 
accepts $w$ if{}f $w \models \G(\psi[\setG])$ for all $\G\psi \in \setG$.
\end{itemize}

\noindent {\bf Subcomponent $\A_\setG$}. The subcomponent is the product of the monitor
for $\varphi[\setG]$ and $\prodaut(\setG)$:
\[\A_\setG  = (\; 2^{Ap} \; , \; {\it Reach}(\varphi[\setG]) \times \prod_{i=1}^n Q_i \; , \; \delta_{\varphi[\setG]} \times \prod_{i=1}^n \delta_i \; , \; \{ \{ \tt \} \times F_1', \ldots, \{ \tt \} \times F_n'\} \; )\]

\begin{itemize}
\item[$\star$] We have: $\A_\setG$ accepts a word $w$ from the state $(\varphi'[\setG], q_{01}, \ldots, q_{0n})$ 
if{}f $w \models \varphi'[\setG] \wedge \G(\psi[\setG])$.
\end{itemize}

\noindent {\bf Connecting $\epsilon$-transitions:} Finally, we describe the $\epsilon$-transitions connecting 
the initial component $\Ain$ to the accepting component $\Aac$.
There is an $\epsilon$-transition for each state $\varphi'$ of $\Ain$ and
each set $\setG= \{ \G\psi_1, \ldots, \G\psi_n\} \subseteq \gsf(\varphi)$. The transition is 
$\big(\varphi', \epsilon, (\, \varphi'[\setG], \G(\psi_1[\setG]), \ldots, \G(\psi_n[\setG])\,)\big)$.
\begin{itemize}
\item[$\star$] {\bf Theorem 1} of \cite{DBLP:conf/cav/SickertEJK16} proves that $w \models \varphi$ if{}f  $\A_\varphi$ accepts $w$. 
\end{itemize}
\noindent This concludes the description of $\A_\varphi$.

\paragraph{Proof of Proposition \ref{prop:translation}.}

We start by generalizing Lemma 5 of \cite{DBLP:conf/cav/SickertEJK16} as follows:

\begin{lemma} \label{lem:prodcorrect2}
$\oracledet(\G\psi')$ accepts a word $w$ from a state $(\xi_1, \xi_2)$ if{}f 
$w \models \G\psi' \wedge \xi_1 \wedge \xi_2$.
\end{lemma}
\begin{proof}
We first claim that there is a formula $\xi$ such that $\xi_1 \wedge \xi_2 \equiv \xi \wedge \psi'$.
If $(\xi_1, \xi_2)=(\psi', \true)$, then $\xi = \true$.
Otherwise there is $(\xi_1', \xi_2')$ and $\nu$ such that $\delta((\xi_1', \xi_2'), \nu) = (\xi_1, \xi_2)$. 
By the definition of the transition function, either $\xi_1 = \eta_1 \wedge \psi'$ for some $\eta_1$,
or $\xi_2 = \eta_2 \wedge \psi'$ for some $\eta_2$, and we can choose $\xi$ accordingly.

By this result, we have $\aft(\xi_1 \wedge \xi_2, \nu) \equiv \aft(\xi \wedge \psi', \nu) \equiv \aft(\xi, \nu) \wedge \aft(\psi', \nu)$ for every state $(\xi_1, \xi_2)$ and letter $\nu$, and so in particular
\begin{equation}
\label{eq1}
\aft(\xi_1, \nu) \wedge \aft(\xi_2, \nu) \models \aft(\psi', \nu)
\end{equation}

We now prove the lemma. Let $v$ be a finite word leading from the initial state $(\psi', \true)$ to $(\xi_1, \xi_2)$. We proceed by induction of the length of $v$.

\noindent {\bf Basis.}  $v = \epsilon$. Then $(\xi_1, \xi_2) = (\psi', \true)$, and so $w \models \G\psi' \wedge \xi_1 \wedge \xi_2$
if{}f $w \models \G\psi' \wedge \psi' \equiv \G\psi'$. By Lemma 5 of \cite{DBLP:conf/cav/SickertEJK16} $w \models \G\psi'$
if{}f $\oracledet(\G\psi')$ accepts $w$ from $(\psi', \true)$, and we are done.

\noindent {\bf Step.} $v = v' \nu$ for some word $v'$ and letter $\nu$. Then there is a state $(\xi_i', \xi_2')$ 
such that $\delta((\psi', \true), v')= (\xi_1', \xi_2')$ and $\delta((\xi_1', \xi_2'), v')= (\xi_1, \xi_2)$.
By induction hypothesis a word $w$ is accepted from $(\xi_1', \xi_2')$ if{}f 
$w \models \G\psi' \wedge \xi_1' \wedge \xi_2'$.
We consider two cases.

If $\aft(\xi_1', \nu) \equiv_P  \true$, then by the definition of $\delta$ we have 
 $\delta((\xi_1', \xi_2'), \nu) = (\aft(\xi_2', \nu) \wedge \psi', \true)$. It follows:

\begin{center}
\begin{tabular}{cl}
       & $\oracledet(\G\psi')$ accepts $w$ from $(\xi_1, \xi_2)$ \\
  if{}f & (determinism) \\
       & $\oracledet(\G\psi')$ accepts $\nu w$ from $(\xi_1', \xi_2')$ \\
  if{}f & (induction hypothesis) \\
       & $\nu w \models \G\psi' \wedge \xi_1' \wedge \xi_2'$ \\
  if{}f & (fundamental property of $\aft$) \\
       & $w \models \aft(\, \G\psi' \wedge \xi_1' \wedge \xi_2' \, , \, \nu)$ \\
  if{}f & (definition of $\aft$) \\
       & $w \models \G\psi' \wedge \aft(\psi' \nu) \wedge \aft(\xi_1', \nu) \wedge \aft(\xi_2', \nu)$ \\
  if{}f & (Equation \ref{eq1}) \\
       & $w \models \G\psi' \wedge \aft(\xi_1', \nu) \wedge \aft(\xi_2', \nu)$ 
\end{tabular}
\end{center}

We conclude the proof by showing
$\G\psi' \wedge \aft(\xi_1', \nu) \wedge \aft(\xi_2', \nu) \equiv \G\psi' \wedge \xi_1 \wedge \xi_2$.
It suffices to prove $\xi_1 \wedge \xi_2 \equiv \aft(\xi_1', \nu) \wedge \aft(\xi_2', \nu) \wedge \psi'$.
Consider two cases:
\begin{itemize}
\item $\aft(\xi_1', \nu) \equiv_P \true$. Then, by the definition of $\oracledet(\G\psi')$, we have
$\xi_1 = \aft(\xi_2', \nu) \wedge \psi'$ and $\xi_2 = \true$, and we are done.
\item $\aft(\xi_1', \nu) \not\equiv_P \true$. Then,by the definition of $\oracledet(\G\psi')$, we have
$\xi_1 = \aft(\xi_1', \nu)$ and $\xi_2 = \aft(\xi_1', \nu) \wedge \psi'$, and we are done. \qed
\end{itemize}
\end{proof}

We can now proceed to prove Proposition \ref{prop:translation}.

\begin{refproposition}{prop:translation}
For every LTL formula $\varphi$, every state $s$ of the LDBA of \cite{DBLP:conf/cav/SickertEJK16} for $\varphi$ can be labelled by an LTL formula $\mathit{label}(s)$ such that (i) $\lang(s)=\lang(\mathit{label}(s))$ and (ii) $\mathit{label}(s)$ is a Boolean combination of subformulas of $\varphi[T_s, F_s]$ for some $T_s$ and $F_s$. Moreover, the LDBA is initial-deterministic.

Further, $\mathit{label}(s)$ can be computed in linear time from the descriptor of $s$.
\end{refproposition}
\begin{proof}
Recall the two properties of the $\aft$ function.
For every formula $\varphi$, finite word $v$, and $\omega$-word $w$:
\begin{itemize}
\item[(i)]  $vw \models \varphi$ if{}f $w \models \aft(\varphi,v)$.
\item[(ii)]  $\aft(\varphi, v)$ is a boolean combination of subformulas of $\varphi$.
Therefore, every formula of ${\it Reach}(\varphi)$ is a boolean combination of subformulas of $\varphi$.
\end{itemize}

Let $s$ be a state of $\Ain$, and let $v$ be any finite 
word leading from $s_0$ to $s$. By Theorem 1 of \cite{DBLP:conf/cav/SickertEJK16}, we have $L(s_0) = L(\varphi)$. 
By (i), $\lang(\aft(\varphi,v)) = \{w \mid vw \in \lang(\varphi) \}$. Since $\Ain$ is deterministic, 
$\lang(s) = \{w \mid vw \in \lang(s_0) \} = \{w \mid vw \in \lang(\varphi) \} = \lang(\aft(\varphi,v))$.
So we can take $\mathit{label}(s) = s$.

We consider now the case that $s$ belongs to $\Aac$. Then there is a set $\setG=(\G\psi_1, \ldots, \G\psi_n)$
such that $s$ belongs to $\A_\setG$. By the definition of $\A_\setG$ as product of DBAs,
$s$ is of the form $(\varphi'[\setG], (\xi_{11}, \xi_{21}), \ldots, (\xi_{1n}, \xi_{2n}))$,
where $\varphi'[\setG]$ is a state of the monitor for $\varphi[\setG]$, and $(\xi_{1i}, \xi_{2i})$ is a state of the 
monitor for $\G(\psi_i[\setG])$.  Further, the words recognized from $s$ are those simultaneously
recognized from $\varphi'[\setG]$, $(\xi_{11}, \xi_{21}), \ldots, (\xi_{1n}, \xi_{2n})$ in their respective automata. 
By Lemma \ref{lem:prodcorrect2}, the words recognized from $s$ are those satisfying 
$\varphi'[\setG] \wedge \xi_{11} \wedge \xi_{21} \wedge \ldots \wedge \xi_{1n} \wedge \xi_{2n}$.
We choose $\mathit{label}(s)$ as this formula. 
It remains to show that each conjunct of $\mathit{label}(s)$ is a boolean combination of formulas of $\mathsf{sf}(\varphi)[\setG]$.

\begin{itemize}
\item By the definition of the monitor for $\varphi[\setG]$, the formula $\varphi'[\setG]$ belongs to 
${\it Reach}(\varphi[\setG])$, and by (ii) we are done.
\item By the definition of the monitor for $\G(\psi_i[\setG])$, the formulas $\xi_{1i}$ and $\xi_{2i}$ belong to 
$\Reach(\psi_i[\setG])$. By (ii), they are boolean combinations of subformulas of $\psi_i[\setG]$.
Since $\psi_i$ is a subformula of $\varphi$, they are also boolean combinations of subformulas of 
$\varphi[\setG]$, and so a boolean combination of formulas of $\mathsf{sf}(\varphi)[\setG]$. \qed
\end{itemize}
\end{proof}

\section{Optimisations}

\subsection{Parallelisation}

Since multiple cores are abundant these days, we use a simple trick to obtain small automata: we launch two threads for $\varphi$ and $\neg \varphi$ and compute the DPA. Observe that complementing parity automata is cheap, since we just need to change the parity of the acceptance condition. We then return the smaller automaton and also prematurely cancel one of the translations, if we already know that the running translation will produce a larger automaton.

\subsection{Reduction of Rankings}

Since we can freely change the order of states we jump to, a good heuristic is to sort jumps to accepting components that are labelled with $\F$-Operators---more precisely, belong to the syntactic class of pure eventual formulas---before any other accepting component. The reasoning is that components for the pure eventual fragment will never be removed from the ranking. Analogously, if an accepting component is only labeled with formulas from the $(\X, a)$-fragment, a safety language is described. These are volatile in the sense that after a fixed number of steps---the nesting depth of the $\X$'s---we know if the obligations are fulfilled or not. Thus we do not need to track all these components. Just one at a time and whenever one of these fails we switch to the next one.

\end{document}